\input epsf
\documentclass[12pt]{article}

\usepackage{graphicx}
\usepackage{subfigure}
\usepackage{bm}
\usepackage{amsmath,amsthm,amssymb}
\usepackage{amsfonts}    
\usepackage{color}      
\usepackage{subfigure}  
\usepackage{epsfig}
\usepackage{float}
\usepackage{epstopdf}
\usepackage{appendix}
\usepackage{multirow}
\usepackage{subfiles}

\newtheorem{Th}{Theorem}
\newtheorem{Def}{Definition}

\newtheorem{lemma}{Lemma}

\usepackage[colorlinks=true, pdfborder=001, linkcolor=blue, anchorcolor=blue, citecolor=blue, urlcolor=blue]{hyperref}

\begin{document}
%================================================================
\renewcommand{\refname}{\begin{flushleft}{{\bf\small
Reference}}\end{flushleft}}
\title {{\bf New solution of Einstein-Yang-Mills equations }
}

\author{Yuewen Chen\thanks{Yau Mathematical Sciences Center, Tsinghua University, Beijing 100084, P.R. China.  E-mail: yuewen\_chern@amss.ac.cn},
\quad Jie Du\thanks{Yau Mathematical Sciences Center, Tsinghua University, Beijing, 100084, P.R. China. Yanqi Lake Beijing Institute of Mathematical Sciences and Applications, Beijing 101408, P.R. China.  E-mail: jdu@tsinghua.edu.cn}, 
\quad Shing-Tung Yau\thanks{Yau Mathematical Sciences Center, Tsinghua University, Beijing, 100084, P.R. China. Yanqi Lake Beijing Institute of Mathematical Sciences and Applications, Beijing 101408, P.R. China. Department of Mathematics, Harvard University, Cambridge, MA 02138, USA. E-mail: yau@math.harvard.edu}
}
\date{}
\maketitle

\centerline{\bf Abstract}
In this paper, we show the numerical solution for spherically symmetric $SU(2)$ Einstein-Yang-Mills (EYM) equations. 
We show the  existence of entropy weak solution for EYM.

\bigskip


\pagenumbering{arabic}

\bigskip

\tableofcontents
%=============================================
\section{Introduction}

The Einstein-Yang-Mills (EYM)  equation plays an important role in GR. In this paper, we aim to find  the stable solution of EYM equations with the $SU(2)$ gauge group numerically.  

\subsection{Einstein-Yang-Mills equations}
We first introduce the formulations and basic properties of EYM equations. We adopt the following static, spherically symmetric metric %\cite{cho2}
$$ g=-A e^{-2 \delta} dt^2 +\frac{dr^2}{A} +r^2 (d\theta^2 +\sin^2\theta d\phi ^2),$$ 
where $(\theta,\phi)$ is the spherical coordinate system, $r$ is the radius, $t$ is the coordinate time, $A=A(r)$, and $\delta=\delta(r)$. 
Denoting $\tau_1,\tau_2,\tau_3$ as the Pauli matrices, the spherically symmetric Yang-Mills connection with $SU(2)$ gauge group can be written in the following form
$$\mathfrak{A}=W(r) \tau_1  d\theta +\big(\cot \theta \tau_3+W(r) \tau_2 \big)\sin\theta d\phi,$$ 

The EYM equations with $SU(2)$ gauge potential have been derived in many papers \cite{s2,Bartnik,cho2,un4}, and take the following form
\begin{align}
	r^2 AW^{''} &=\Big(\frac{(W^2-1)^2}{r}+r(A-1)\Big)W^{'}+W(W^2-1)  \label{eq:4},\\
	r A^{'} &=1-\frac{(W^2-1)^2}{r^2}-A(2{W^{'}}^2+1)  \label{eq:5}, \\
	\delta^{'}&=-\frac{2(W^{'})^2}{r}\label{eq:3A},
\end{align}
with  boundary conditions
\begin{align}
	W(0) &=\pm 1,\quad W(\infty) =\mp 1,\\
	A(0) &=1,\quad \delta(0) =0. 
	\end{align}
In this system, Eq. (\ref{eq:4}) is the matter field equation for solving $W(r)$ in the Yang-Mills field, also called the Yang-Mills equation. Eqs. (\ref{eq:5}) and (\ref{eq:3A}) are Einstein equations to determine $A(r)$ and $\delta(r)$ in the metric, in which Eq. (\ref{eq:5}) is also called the Hamiltonian constraint equation. Notice that Eqs. (\ref{eq:4}) and (\ref{eq:5}) do not involve $\delta$, and hence one can first solve these two equations for $A(r)$ and $W(r)$ and then use (\ref{eq:3A}) to obtain $\delta(r)$. If the solution to the EYM equations satisfies $A(r_*)=0$ at some point $r_*$, then we call such a solution as a black hole solution and $r_*$ is the position of the event horizon. 
%--------------------------------------

In 1988, Bartnik and McKinnon \cite{Bartnik} numerically discovered a global nontrivial static nonsingular particle-like solution that is not a black hole solution; this work sparked a great deal of interest in the general relativity community\cite{2, 3, Kun}.

The EYM equation has also caught the interest of specialists in the field of differential equations based on the numerical observation. J. Smoller and his associates published a number of publications \cite{s1,s2,s3,s4,s5}that represented the key advancements in the theoretical analysis. The $SU(2)$ EYM equations accept an infinite family of black hole solutions with a regular event horizon, as demonstrated by Smoller and Yau {\it et.al}, who conclusively demonstrated the existence of a globally defined smooth static solution . 
In the meanwhile, they established that there are an endless number of smooth, static, regular solutions to the EYM equations \cite{s1}.

 However, all solutions founded in history are not dynamically stable and therefore are not physical  \cite{un1,un2,un3,un4}.   In this paper, we present the high order schemes to solve EYM equations globally and show  a stable solution for EYM.

In this paper, we are interested in finding a stable black hole solution globally. Since the horizon $r_*$ is usually close to $r=0$ and more details of the solution concentrate in the region with a small $r$, we adopt the following coordinate transformation 
\begin{align}
	x&=\ln(r) \label{eq:t},\quad r=e^{x},
\end{align}
and then rewrite the equations (\ref{eq:4}) and (\ref{eq:5}) for $W$ and $A$ as 
\begin{align}
	c(r,A,W)W_x=A W_{xx}+(1-W^2)W,\label{eq:7}\\
	A_x=-(1+\frac{2}{r^2}W_x^2) A+1-\frac{1}{r^2}(W^2-1)^2, \label{eq:8}
\end{align}
where 
$$c(r,A,W)=2A-1+\frac{1}{r^2}(W^2-1)^2.$$
The new idea in this paper, which is different from classical methods, is to consider the steady state of the parabolic version of EYM equations. Instead of solving the above static system of equations (\ref{eq:7}) and (\ref{eq:8}) directly, we raise the problem one-dimensional higher and consider the following time-dependent parabolic problem
\begin{align}
	W_t&+c(r,A,W)W_x =A W_{xx}+(1-W^2)W,\label{eq:9}\\
	A_x&=-(1+\frac{2}{r^2} W_x^2) A+1-\frac{1}{r^2}(W^2-1)^2  ,\label{eq:10}
\end{align}
where $A=A(x,t)$ and $W=W(x,t)$. By introducing suitable initial conditions, we aim to march this system in time to steady state numerically. We will discuss more details about the choices of the initial condition in the numerical examples section.

\subsection{High order WENO scheme}

Since only shooting methods were used to solve the static EYM equations directly and no one considers its time evolution version (\ref{eq:9})-(\ref{eq:10}) in literature, it is meaningful to investigate more suitable numerical methods. This problem forms a  convection-diffusion system with source terms. Since the diffusion coefficient in (\ref{eq:9}) is $A$, the system becomes degenerate for black hole solutions in which $A=0$ at the horizon. In this case, there will be a sharp front in the solution near the horizon. 

 In this paper, we are interested in solving this system by using  high order WENO methods.

\subsection{Contributions and organization of the paper}

The rest of this paper is organized as follows. In section 2 we give the definition of entropy weak solution of EYM equations and the jump condition and the convergence analysis of first order TVD  finite difference scheme.
In section 3, to make it simpler to compute the implicit scheme, we construct a new WENO scheme for the second derivative of the YM equation. In section 4, we construct a WENO scheme for constraint equation. In section 5,  we provide  numerical test to demonstrate the behavior of the  new WENO scheme and the numerical solution of EYM systems. Finally, we will give the conclusions in Section 6.

\section{ First order finite difference scheme and convergence analysis}

In this section, we would like to  study the jump condition of EYM and first order scheme.
Firstly, we define the weak solution of EYM equation that belong to BV class and satisfy the entropy condition and give the RH jump condition.
Secondly, we study the convergence of first-order schemes and prove that it  can converge to weak solutions.Moreover ,such  weak solution  also satisfy the entropy condition.

There is no a prior estimate for $A$ to guarantee that $A \geq 0$ during the evolution,  to maintain the parabolic properties, we define 
$\tilde{A}=\max(0,A)$. However, by numerical experience, in the steady sate, we show that $A \geq 0$, then $A=\tilde{A}$ (See Figure \ref{fig9}). These guarantee the steady state solution of system (\ref{eq:18}) satisfy the static EYM equation.

\subsection{Entropy weak solution and jump condition }

So, we are interested in the following modified  problem

\begin{align*}%\label{eq:18}
    A_x&=-A(1+\frac{2}{r^2}W_x^2)+1-\frac{1}{r^2}(1-W^2)^2,\\
    \tilde{A}&=\max(A,0),\\
    W_t+B(x,W) W_x &=\tilde{A} W_{xx}+(1-W^2)W ,
    \end{align*}
where $$ B(x,W)=2A-1+\frac{1}{r^2}(1-W^2)^2,$$
and we denote 
 $g(W)=W(1-W^2)$, $a(x)=2A-1 $ and $\displaystyle f(W)=\int^W (1-s^2)^2 \ ds=W(1-\frac{2}{3}W^2+\frac{1}{5}W^4)$.
We can rewrite $B(x,W)W_x$ as 
\begin{align*}
    B(x,W)W_x&=(2A-1)W_x+\frac{1}{r^2}f(W)_x\\
    &=a(x)W_x+\frac{1}{r^2}f(W)_x.
\end{align*}
%========================
\begin{align}\label{eq:18}
    A_x&=-A(1+\frac{2}{r^2}W_x^2)+1-\frac{1}{r^2}(1-W^2)^2,\\
    \tilde{A}&=\max(A,0),\\
    W_t+a(x)W_x+\frac{1}{r^2}f(W)_x &=\tilde{A} W_{xx}+(1-W^2)W .\label{eq1113}
\end{align}

Following Wu and Yin \cite{Wu},
we give the definition of entropy weak solution for EYM equations. Let $Q_T=\{(t,x): 0<t<T,x\in[-5,5]\}$. 
\begin{Def}
A  function $W(x,t)\in L^\infty(Q_T)\bigcap BV(Q_T)$ is said to be an entropy weak solution of Eq. (\ref{eq1113}) if the following two conditions  holds:\\
(1) $A$ satisfy constraint equation (18).\\
(2) For any $ c \in \mathbb{R}$ and $\phi \in C_0^\infty(\mathbb{R},\mathbb{R}^+) $,
\begin{align}\label{eq:entro2}
  \iint_{Q_T} |W-c|\phi_t +(a(x)\phi)_x |W-c|+(\frac{1}{r^2}\phi)_x sgn(W-c)(f(W)-f(c))\nonumber\\
-(\tilde{A}\phi)_x|W-c|_x  +sgn(W-c) g(W) \phi \ dx dt \geq 0.  
\end{align}

\end{Def}

For more theory about scalar degenerate convection diffusion equation, one can see \cite{Wu}. Numerically, we are interested in the shock speed.
We derive the jump condition of $W$. Since $A$ is  Lipschitz continue, there is no jump condition for $A$.

 \textbf{Remark} Since Eq. (18) is a standard linear ODE, which can be solved as
 \begin{align*}
     A=e^{Q(-5)-Q(x)}+e^{-Q(x)}\int_{-5}^x (1-\frac{1}{r^2}(1-W^2)^2)e^{Q(s)} ds,
 \end{align*}
 where $$Q(x)=\int_{-5}^x 1+\frac{2}{r(s)^2}W_x^2(s) ds .$$

 We will give a jump condition for EYM equation. Let $\Gamma$ be a smooth curve across which $W$ has a jump discontinuity,where $\Gamma$ is given by $x=x(t)$. The shock speed of the discontinuity is given by $\displaystyle s=\frac{dx}{dt}$. Assume $P$ is any point on $\Gamma$, and $D$ be a small ball centered at $P.$ Let
 $\phi \in C_0^ \infty(D)$,then we have  the follow jump condition.
\begin{Th}
The shock speed $s$ is given by
$$ s[W]=\frac{1}{r^2}[f]+a[W]-\{\tilde{A}_x\}[W]-[\tilde{A}W_x],$$
where $[W]=W^+-W^-, \{ \tilde{A}_x \}=\frac{1}{2}(\tilde{A}_x^+ +\tilde{A}_x^-)$.
Moreover, in static case, 
we have the RH jump condition for static EYM black hole solution
$$ 0=\frac{1}{r^2}\frac{[f]}{[W]}+a-\{{A}_x\}-\frac{[{A}W_x]}{[W]}.$$
\end{Th}

\textbf{Remark} As $t \to \infty$,$A \geq0$, then $\tilde{A}=A$. However, it is difficult to prove this phenomena. By lots of numerical experience, we find that $\tilde{A}=A$ in steady state (One can see Fig. \ref{fig9} for more detail). So we get a jump condition for static EYM black hole solution
$$0= a+\frac{1}{r^2}\frac{[f]}{[W]}-\frac{[AW_x]}{[W]}-\{ A_x\}.$$

\subsection{Convergence analysis for TVD scheme}

In this section, we will construct the first order TVD scheme and  prove that the approximate solution could converge to the  entropy weak  solution.

We set up a grid: let $-5=x_0 <x_1<...<x_N=5$ denote a uniform grid with a mesh size $h=x_2-x_1.$
In time direction, $t=n \  dt,n\in \mathbb{Z}^+$. For a function $W(x,t)$, we use the notation $W^n_i$ to denote the value of $W$
at the mesh point $(i\  h, n \ dt).$ The  discrete $L^\infty$ norm ,the $L^1$ norm and the $BV$ seminorm are defined as follows:
\begin{align*}
    \|W\|_{\infty}&=\max_{0\leq i \leq N} |W_i|,\\
    \|W\|_{1}&=h\sum_{0\leq i \leq N} |W_i|,\\
    |W|_{BV}&=\sum_{0\leq i \leq N-1}|W_{i+1}-W_i|.
\end{align*}
Then a Lax-Friedrichs scheme for YM equation is given by

\begin{align} 
 A^n_{i+1}-A^n_i+hA^n_{i+1}(1+\frac{2}{r_i^2}\frac{(W^n_{i+1}-W^n_i)^2}{h^2}) &=\frac{h}{2}(F^n_{i+1}+F^n_i),\label{eq:24}\\
 \tilde{A}^n_i&=\max(A^n_i,0),\\
 \frac{1}{dt}(W_i^{n+1}-W^n_i)+\frac{a^n_i}{2h}(W^n_{i+1}-W^n_{i-1})+\frac{1}{r_i^2}\frac{1}{2h}(f^n_{i+1}-f^n_{i-1})  &=\frac{\alpha^n_{i+\frac{1}{2}}}{2h}(W^n_{i+1}-W^n_{i})\nonumber\\
 -\frac{\alpha^n_{i-\frac{1}{2}}}{2h}(W^n_{i}-W^n_{i-1})
 +\tilde{A}^n_i\frac{1}{h^2}(W^n_{i+1}-2W^n_i+W^n_{i-1})+g^n_i \label{eq136},
 \end{align}

where
$$ F=1-\frac{1}{r^2}(1-W^2)^2,$$
and the viscosity coefficient $\alpha_{i+\frac{1}{2}}$ can be chosen as 
 Lax-Friedrichs type 
\begin{align*}
    \alpha_{i+\frac{1}{2}} =\max_i( |a_{i}|+\frac{1}{r_i^2}(1-W^2_{i})^2 ).
\end{align*}

Next, we would present a-prior estimate  for scheme 
(\ref{eq136}).
\begin{lemma}
Under following CFL condition

\begin{equation}
    \frac{1}{2}>\frac{dt}{h}\max_i|\alpha^n_{i+\frac{1}{2}}|+\frac{2dt}{h^2}\max_i(|A^n_i|),dt<\frac{1}{4}. \label{eq137}
\end{equation}
We have 
$$\|W^{n+1}\|_{\infty} \leq 1.$$
\end{lemma}

\begin{proof}
      \begin{align*}
        W^{n+1}_i&=W^n_i-\frac{dt}{2h}a^n_i(W^n_{i+1}-W^n_{i-1})-\frac{1}{r_i^2}\frac{dt}{2h}f^{'}(\theta_i)(W^n_{i+1}-W^n_{i-1})\\
        &+\frac{dt}{2h}\alpha^n_{i+\frac{1}{2}}(W^n_{i+1}-W^n_{i}) -\frac{dt}{2h}\alpha^n_{i-\frac{1}{2}}(W^n_{i}-W^n_{i-1})\\
        &+\frac{dt}{h^2}\tilde{A}^n_i(W^n_{i+1}-2W^n_i+W^n_{i-1})+dtg^n_i\\
        &=W_i^n(\frac{1}{2}-\frac{dt}{2h}\alpha^n_{i+\frac{1}{2}} -\frac{dt}{2h}\alpha^n_{i-\frac{1}{2}}-\frac{2dt}{h^2}\tilde{A}^n_i)\\
        &+W^n_{i+1}(-\frac{dt}{2h}a^n_i-\frac{1}{r_i^2}\frac{dt}{2h}f^{'}(\theta^n_i)+\frac{dt}{2h}\alpha^n_{i+\frac{1}{2}}+\frac{dt}{h^2}\tilde{A}^n_i)\\
        &+W^n_{i-1}(\frac{dt}{2h}a^n_i+\frac{1}{r_i^2}\frac{dt}{2h}f^{'}(\theta^n_i)+\frac{dt}{2h}\alpha^n_{i-\frac{1}{2}}+\frac{dt}{h^2}\tilde{A}^n_i)+\frac{1}{2}W^n_i+dt g^n_i,
    \end{align*}
    where $\theta^n_i$ is between $W^n_{i-1}$ and $W^n_{i+1}$.
   Assume that $\|W^n\|_{\infty}\leq 1$, Using Lemma \ref{Lemma2}, we have 
    $$|\frac{1}{2}W^n_i+dt g^n_i|\leq \frac{1}{2}.$$
    Under the CFL condition, all the coefficients of $W^n_{i-1},W^n_i,W^n_{i+1}$ are nonnegative, then
    \begin{align*}
        |W^{n+1}_i| &\leq |W^n_i||(\frac{1}{2}-\frac{dt}{2h}\alpha^n_{i+\frac{1}{2}} -\frac{dt}{2h}\alpha^n_{i-\frac{1}{2}}-\frac{2dt}{h^2}\tilde{A}_i) |\\
        &+|W^n_{i-1}||\frac{dt}{2h}a^n_i+\frac{1}{r_i^2}\frac{dt}{2h}f^{'}(\theta^n_i)+\frac{dt}{2h}\alpha^n_{i-\frac{1}{2}}+\frac{dt}{h^2}\tilde{A}_i |\\
        &+|W^n_{i+1}||-\frac{dt}{2h}a^n_i-\frac{1}{r_i^2}\frac{dt}{2h}f^{'}(\theta^n_i)+\frac{dt}{2h}\alpha^n_{i+\frac{1}{2}}+\frac{dt}{h^2}\tilde{A}_i |+\frac{1}{2}\\
        &\leq \frac{1}{2}\max(|W^n_{i-1}|,|W^n_i|,|W^n_{i+1}|)+\frac{1}{2}.
    \end{align*}
    Then
    $$\|W^{n+1}\|_{{\infty}} \leq 1.$$
    
\end{proof}

\begin{lemma}\label{Lemma2}
    Let $Q(W)=\frac{1}{2}W+dt W(1-W^2)$, if $|W| \leq 1,dt<\frac{1}{4}$, then $$Q'(W)>0,$$
    and $$ |Q(W)|\leq \frac{1}{2},$$
    i.e. for any $W_{i+1} \geq W_i$, $|W_{i+1}|\leq1,|W_i|\leq 1,$ then
    $$Q(W_{i+1}) \geq Q({W_i}).$$
\end{lemma}
\begin{proof}
     By direct calculation, we get
    $$Q'(W)=\frac{1}{2}+dt(1-3W^2).$$
    Since $|1-3W^2|\leq 2$, if we take
    $\frac{1}{4}>dt$, then
    $$Q'(W)=\frac{1}{2}+dt(1-3W^2)>0.$$
    Then $Q(W)$  attained the maximum  $Q(1)=\frac{1}{2}$ and the minimum  $Q(-1)=-\frac{1}{2}$, so
    $$|Q(W)|\leq \frac{1}{2}.$$
    Then for $W_{i+1} \geq W_{i}$, we get
    $$Q(W_{i+1}) \geq Q(W_i).$$
    
\end{proof}

A large number of numerical experiments show that for any initial  conditions, the obtained steady-state solutions are monotonic, so we only need to study the monotonic initial conditions. For the monotonic initial conditions, we can get the following $BV$ estimate.
 
\begin{lemma}
Assume  $W^n$ is monotonically increasing, then
    $$|W^{n+1}|_{BV} \leq |W^n|_{BV}.$$
\end{lemma}

\begin{proof}
    We rewrite the scheme  (\ref{eq136}) as Harten's version
   \begin{align*}
       W_i^{n+1}&=W_i^n+D_{i+\frac{1}{2}}(W^n_{i+1}-W_i^n)-C_{i-\frac{1}{2}}(W_i^n-W^n_{i-1})+dt g^n_i\\
       &=\frac{1}{2}W_i^n+D_{i+\frac{1}{2}}(W^n_{i+1}-W_i^n)-C_{i-\frac{1}{2}}(W_i^n-W^n_{i-1})+\frac{1}{2}W_i^n+dt g^n_i,
   \end{align*} 
   where
   \begin{align*}
       D_{i+\frac{1}{2}}&=\frac{dt}{h}(-\frac{1}{2}a_i^n-\frac{f'(\theta^n_i)}{2r_i^2}+\frac{1}{2}\alpha^n_{i+\frac{1}{2}})+\frac{dt}{h^2}\tilde{A}_i,\\
       C_{i-\frac{1}{2}}&=\frac{dt}{h}(+\frac{1}{2}a_i^n-\frac{f'(\theta^n_i)}{2r_i^2}+\frac{1}{2}\alpha^n_{i-\frac{1}{2}})+\frac{dt}{h^2}\tilde{A}_i,
   \end{align*}
   and $\theta^n_i$ is between $W^n_{i-1}$ and $W^n_{i+1}$.
   Using Lemma \ref{Lemma2}, the term $\frac{1}{2}W_i^n+dt g_i$ is monotone increase with respect $W_i$, we define $Q_i=\frac{1}{2}W_i^n+dt g_i$.
   Taking $W^{n+1}_{i+1}-W^{n}_i$ and sum, we have
   \begin{align*}
       W^{n+1}_{i+1}-W^{n+1}_i&=\frac{1}{2}(W^n_{i+1}-W^n_i)+D_{i+\frac{3}{2}}(W^n_{i+2}-W^n_{i+1})-C_{i+\frac{1}{2}}(W^n_{i+1}-W_i^n)+Q_{i+1}\\
       &-D_{i+\frac{1}{2}}(W^n_{i+1}-W^n_i)+C_{i-\frac{1}{2}}(W^n_i-W^n_{i-1})-Q_i,
   \end{align*}
   and
   \begin{align*}
       \sum_i |W^{n+1}_{i+1}-W^{n+1}_i|
       &\leq \sum_i (\frac{1}{2}-C_{i+\frac{1}{2}}-D_{i+\frac{1}{2}})|W^n_{i+1}-W_i^n|\\
       &+\sum_i C_{i-\frac{1}{2}}|W^n_i-W^n_{i-1}|\\
       &+\sum_i D_{i+\frac{3}{2}}|W^n_{i+2}-W^n_{i+1}|\\
       &+\sum_i (Q_{i+1}-Q_i).
   \end{align*}
   Since $Q_i$ is monotone increase with respect $W_i$,
   we have $$\sum_i Q_{i+1}-Q_i=(Q_N-Q_0)=1,$$
   then
   $$\sum_i |W^{n+1}_{i+1}-W^{n+1}_i|\leq \frac{1}{2}|W^n|_{BV}+1 \leq 2.$$
\end{proof}
 We need a Lemma to  show the $L_1$ continue in time direction. That is 
\begin{lemma}
    
    $$ \| W^{m}-W^{n}\|_{1} \leq \sqrt{(m-n) \Delta t}.$$
\end{lemma}

\begin{proof}
    We rewrite the scheme (\ref{eq136}) as
    \begin{align} \label{eq47}
        W_i^{n+1}-W^n_i &=\frac{dt}{h^2}\tilde{A}_i(W^n_{i+1}-2W^n_i+W^n_{i-1})+\frac{dt}{h}D_{i+1/2}(W^n_{i+1}-W^n_i)\nonumber\\
        &-\frac{dt}{h}C_{i-1/2}(W^n_{i}-W^n_{i-1})+dt g^n_i,
     \end{align}
     where 
    \begin{align*}
        D_{i+1/2} &=-\frac{1}{2} a^n_i -\frac{1}{2 r_i^2} f^{'}(\theta^n_i) +\frac{1}{2}\alpha^n_{i+1/2},\\
       C_{i-1/2} &=+\frac{1}{2} a^n_i +\frac{1}{2 r_i^2} f^{'}(\theta^n_i) +\frac{1}{2}\alpha^n_{i-1/2},\\
    \end{align*} 
    and $\theta_i^n$ is between $W^n_{i-1}$ and $W^n_{i+1}$.
Multiplying the difference equation by test function $\phi(x) \in C_0^\infty([-5,5])$ and summation by parts, we have 
\begin{align*}
    h\sum_{i } \phi_i (W^m_i-W^n_i) &= h \sum_{i } \sum_{\ell=n}^{m-1} \phi_i (W^{\ell-1}_i-W^{\ell}_i)\\
    &=h dt  \sum_{i } \sum_{\ell=n}^{m-1} \phi_i D_{i+1/2} \frac{1}{h}(W^{\ell+1}_{i+1}-W^{\ell}_i) - \phi_i C_{i-1/2} \frac{1}{h}(W^{\ell+1}_{i}-W^{\ell}_{i-1})\\
    &+dt h  \sum_{i } \sum_{\ell=n}^{m-1} g^\ell _i \phi_i \\
    &+ dt h \sum_{i } \sum_{\ell=n}^{m-1} \frac{1}{h^2} \tilde{A_i}\phi_i (W^\ell_{i+1}-W^\ell_i) - \frac{1}{h^2} \tilde{A_i}\phi_i (W^\ell_{i}-W^\ell_{i-1}) \\
    &= dt h \sum_{i } \sum_{\ell=n}^{m-1} (\phi_iD_{i+1/2}-\phi_{i+1}C_{i+1/2})\frac{1}{h}(W_{i+1}-W_i)\\
    &+ h dt \sum_{i } \sum_{\ell=n}^{m-1} \frac{1}{h}(\phi_i \tilde{A}_i-\tilde{A}_{i+1}\phi_{i+1})(W^\ell_{i+1}-W^\ell_i)\frac{1}{h}\\
    &+ h dt \sum_{i } \sum_{\ell=n}^{m-1}  g_i^\ell \phi_i,
\end{align*}
where
\begin{align*}
    \frac{1}{h}|\phi_{i+1}\tilde{A}_{i+1}-\phi_i \tilde{A}_i| &=\frac{1}{h}|(\tilde{A}_{i+1}-\tilde{A}_i)\phi+\tilde{A}_i(\phi_{i+1}-\phi_i)|\\
    &\leq \frac{1}{h}|\tilde{A}_{i+1}-\tilde{A}_i||\phi_i|+ |\tilde{A}_i|\frac{1}{h}|\phi_{i+1}-\phi_i|\\
    &\leq C_1(\|\phi\|_{\infty}+\|\phi_x\|_\infty).
\end{align*}
Here we use the $\|A_x\|_{\infty} \leq  +\infty$ (since $A$ is Lipschitz continous)
and 

\begin{align*}
    |\phi_iD_{i+1/2}-\phi_{i+1}C_{i+1/2}| &\leq |\phi_{i}-\phi_{i+1}||D_{i+1/2}|+|\phi_{i+1}||D_{i+1/2}-C_{i+1/2}|\\
    &\leq C_2.
\end{align*}

Then we have
\begin{align*}
 h\sum_{i } |\phi_i (W^m_i-W^n_i)| &\leq  dt  \sum_{\ell=n}^{m-1} \sum_{i } 
 |\phi_i D_{i+1/2}-\phi_{i+1}C_{i+1/2}||W_{i+1}^\ell-W_i^\ell|\\
 &+ dt \sum_{\ell=n}^{m-1} \sum_{i } \frac{1}{h}|\tilde{A}_{i+1}\phi_{i+1}-\phi_i \tilde{A}_i||W^\ell_{i+1}-W^\ell_i|+dtC\\
 &\leq dt (C_2+C_1) \sum_{\ell=n}^{m-1} |W^\ell|_{BV}+dtC\\
 &\leq  (C_2+C_1(\|\phi\|_{\infty}+\|\phi_x\|_\infty ))dt(m-n) |W^n|_{BV}+dtC.
\end{align*}

Next, introduce the function
\begin{equation}
	 \beta(x)=\left\{
	\begin{aligned}
        sgn(\sum_{i \in \mathbb{Z}} (W^m_i-W^n_i)\chi_i(x) )&, \ if \ |x|\leq J-\rho \\
				0 & , otherwise 
	\end{aligned}
	\right.
\end{equation}
where $\chi_i(x)$  is  the characteristic function of $[x_0+ih,x_0+(i+1)h )$ and $J\in \mathbb{Z}$.
Let $\omega_\rho(x)$ be a standard $C_0^\infty$ mollifier given by
$\omega_\rho(x)=\frac{1}{\rho}\omega(\frac{x}{\rho})$, 
where 
\begin{equation}
	 \omega(x)=\left\{
	\begin{aligned}
      \frac{1}{\Omega}exp(\frac{1}{|x|^2-1}),&\ |x|<1\\
      0,&\ |x|\geq 1
	\end{aligned}
	\right.
\end{equation}
and $\displaystyle \Omega=\int_0^1 exp(\frac{1}{|x|^2-1}) dx $.

Let $\beta^\rho =\omega_\rho *\beta$,
we can check that 
$$\|\beta^\rho \|_{L^\infty} \leq 1, \quad \| \beta^\rho_x\|_{L^\infty} \leq O(\frac{1}{\rho}).$$
Taking test function $\phi=\beta^\rho$,  
we get
\begin{align*}
    h \sum_{i=-J}^J |W_i^m-W_i^n| &\leq ((C_3+C_1)+\frac{C_4}{\rho}) dt (m-n) |W^n|_{BV}+ dt C\\
    &\leq \frac{C}{\rho}dt(m-n)|W^n|_{BV}.
\end{align*}
Taking $\rho=\sqrt{(m-n)dt}$ and $J \to \infty $, we have
$$h \sum_{i \in \mathbb{Z}} |W^m_i-W^n_i| \leq C \sqrt{(m-n)dt}.$$
\end{proof}

Next, we will drive a cell entropy inequality. By the standards process in \cite{Majda} and \cite{Evje}.
We use  the standard notations $u \vee v=max(u,v) $, $u \wedge v =min(u,v)$. To simplify the notation, we define the finite difference operators,
$$D^-W_i=\frac{1}{h}(W_i-W_{i-1}),\qquad D^+W_i=\frac{1}{h}(W_{i+1}-W_i).$$

Let $U(W)=|W-c|,\   F(W)=sgn(W-c)(f(W)-f(c))$, where $c$ is a constant. Then the following inequality holds. 
\begin{lemma}
    The cell entropy inequality holds for EYM equation
    \begin{align}\label{eq:entro3}
        \frac{1}{dt}(U^{n+1}_i-U^n_i)+\frac{a_i^n}{2h}(U^n_{i+1}-U^n_{i-1})+\frac{1}{2 h r_i^2}(F^n_{i+1}-F^n_{i-1})\nonumber\\
        -(\tilde{A}^n_i\frac{1}{h^2}+\frac{\alpha^n}{2h})(U^n_{i+1}-2U^n_i+U^n_{i-1})\nonumber\\
        -g^n_i sgn(W^n_i-c) \leq 0.
    \end{align}

\end{lemma}

\begin{proof}
    To simplify the proof, taking the viscosity coefficient as $\alpha^n$
we  rewrite the  scheme as

$$\frac{1}{dt}(W^{n+1}_i-W^n_i)+D^-\Phi(W^n_{i+1},W^n_{i})-\tilde{A}_{i}D^-( D^+W^n_i)-g^n_i =0,$$
 where $\Phi$ is given by
 \begin{align*}
 \Phi(W_{i+1},W_i)&=\frac{ a_i}{2}(W_i+W_{i+1}) +\frac{1}{2}\frac{1}{r_i^2}(f_{i}+f_{i+1})-\frac{{\alpha^n}}{2}(W_{i+1}-W_i).
 \end{align*}
Define
$$H(W^n_{i-1},W^n_{i},W^n_{i+1})=W^n_i-dt D^-\Phi(W^n_{i+1},W^n_{i})+dt \tilde{A}_{i} D^-( D^+W^n_i) +dt {g}^n_i. $$
 It is easy to check  that if CFL condition is satisfied, then  $\frac{\partial H}{\partial W_j} \geq 0, \ j=i-1,\ i,\ i+1$.

Consider 
\begin{align*}
    &H(c \vee W_{i-1} ,c \vee W_i,c \vee W_{i+1})\\
    &=W_i \vee c -dt D^-\Phi(W_{i+1} \vee c,W_{i} \vee c)+dt \tilde{A}_{i}D^-( D^+(W_i\vee c)), \\
\end{align*}
then
\begin{align}\label{eq:entropy}
    &H(c \vee W^n_{i-1},c \vee W^n_i,c \vee W^n_{i+1})-H(c\wedge W^n_{i-1},c\wedge W^n_i,c\wedge W^n_{i+1})\nonumber\\
    &=|W_i-c|- dt D^-(\Phi(W_{i+1} \vee c,W_{i} \vee c)-\Phi(W_{i+1} \wedge c,W_i \wedge c))+dt \tilde{A}_{i} D^-( D^+(W_i\vee c) \nonumber\\
    &-  D^+(W_i\wedge c ))+dt\   sgn(W_i-c) g_i.
\end{align}
By monstrosity of the scheme, we get 
\begin{align*}
 &\  H(c \vee W^n_{i-1},c \vee W^n_i,c \vee W^n_{i+1})-H(c\wedge W^n_{i-1},c\wedge W^n_i,c\wedge W^n_{i+1})\\
 &\geq H(W_{i-1}^n,W_i^n,W_{i+1}^n)\vee c- H(W_{i-1}^n,W_i^n,W_{i+1}^n)\wedge c\\
 &=|W^{n+1}-c|,
\end{align*}
which inserted into Eq. (\ref{eq:entropy}), we get the cell entropy inequality
    \begin{align}\label{eq:38}
    &\frac{|W^{n+1}_i-c|-|W^{n}_i-c|}{dt}\nonumber\\
    &-\tilde{A}_{i}D^-(  D^+(W_i^n\vee c)- D^+(W_i^n\wedge c)\nonumber\\
    &+D^-(\Phi(W^n_{i+1} \vee c,W^n_{i} \vee c)-\Phi(W_{i+1}^n \wedge c,W_{i}^n \wedge c))\nonumber\\
    &-sgn(W^{n}_j-c){g}(W_i^n)\leq 0.
\end{align}
Using the following relation
\begin{align*}
    f(W\vee c)-f(W\wedge c) &=sgn(W-c)(f(W)-f(c))=F(W),\\
    W\vee c-W\wedge c &=sgn(W-c)(W-c)=U(W),
\end{align*}
we get 
\begin{align*}
    &\ \Phi(W_{i+1} \vee c,W_i \vee c)-\Phi(W_{i+1} \wedge c,W_i \wedge c)\\
    &=\frac{a_i^n}{2}(U_i^n+U^n_{i+1})+\frac{1}{2r_i^2}(F^n_{i}+F^n_{i+1})-\frac{\alpha^n}{2}(U^n_{i+1}-U^n_i),
\end{align*}
then, equation (\ref{eq:38}) can be rewritten as

   \begin{align*}
        \frac{1}{dt}(U^{n+1}_i-U^n_i)+\frac{a_i^n}{2h}(U^n_{i+1}-U^n_{i-1})+\frac{1}{2 h r_i^2}(F^n_{i+1}-F^n_{i-1})\\
        -(\tilde{A}^n_i\frac{1}{h^2}+\frac{\alpha^n}{2h})(U^n_{i+1}-2U^n_i+U^n_{i-1})\\
        -g^n_i sgn(W^n_i-c) \leq 0.
    \end{align*}
\end{proof}

Define $(A_\Delta,W_\Delta)$ be the interpolating of degree one using $A^n_i$ and $W^n_i$, where $\Delta=(h,dt)$.
$W_\Delta$ interpolate at the vertices of each rectangle
$$D^n_i=[x_0+i h,x_0+(i+1)h] \times [n dt,(n+1) dt]$$
and $A_\Delta$ interpolate at the one dimension domain $[x_0+ih,x_0+(i+1)h]$.
 $(A_\Delta,W_\Delta)$ are piece wise line segment, and we have
\begin{align*}
    W_\Delta(x,t)&=W_i^n+(W^n_{i+1}-W^n_i)\frac{x-ih}{h}+(W^{n+1}_i-W^n_i)\frac{t-n dt}{dt},\\
    &+(W^{n+1}_{i+1}-W^{n+1}_i-W^{n}_{i+1}+W^n_i)\frac{x-i h}{h}\frac{t-n dt}{dt},\\
    A_\Delta(x,t_n)&=A_i^n+(A^n_{i+1}-A^n_i)\frac{x-ih}{h}. 
\end{align*}

\begin{lemma}
    Let $\{\Delta\}$ be a sequence of democratization parameters tending to zeros. Then there exist a subsequence $\{\Delta_i\}$ such that $\{W_{\Delta_i}\}$ converges in $L^1_{loc}(Q_T)$ and point-wise almost everywhere in $Q_T$ to a limit $W$ as $i \to \infty$.
\end{lemma}

\begin{proof}
    By Lemma 2, we have $\|W_\Delta\|_{\infty} \leq 1$. Using Lemma 3, we have 
\begin{align*}
    \int_{Q_T}|\partial_x W_\Delta| dx dt&\leq \sum_{i,n}\int_{D^n_i}\frac{1}{h}(1-\frac{t-n dt}{dt})|W_{i+1}^n-W^n_i| dxdt\\
    &+\sum_{i,n}\int_{D^n_i} \frac{1}{h}(\frac{t-n dt}{dt})|W_{i+1}^{n+1}-W^{n+1}_i| dxdt\\
    &\leq \frac{dt}{2}\sum_{i,n}|W_{i+1}^n-W^n_i|+\frac{dt}{2}\sum_{i,n}|W_{i+1}^{n+1}-W^{n+1}_i|\\
    &\leq T|W^0|_{BV}.
\end{align*}
Using Lemma 5, we have
\begin{align*}
    \int_{Q_T}|\partial_t W_\Delta| dx dt &\leq  \sum_{i,n}\int_{D^n_i}  \frac{1}{dt}(1-\frac{x-ih}{h})|W^{n+1}_i-W^n_i| dxdt\\
    &+\sum_{i,n}\int_{D^n_i} \frac{1}{dt}\frac{x-i h}{h}|W^{n+1}_{i+1}-W^n_{i+1}| dx dt\\
    &\leq \frac{h}{2} \sum_{i,n}|W^{n+1}_i-W^n_i|+\frac{h}{2} \sum_{i,n}|W^{n+1}_{i+1}-W^n_{i+1}|\\
    &\leq h \sqrt{T} |W^0|_{BV}.
\end{align*}
Then, there is a finite constant $C(T,|W^0|_{BV})>0$ such that 
$$\|W_\Delta\|_{L^\infty}(Q_T) \leq 1,\qquad |W_\Delta|_{BV(Q_T)} \leq C(T,|W^0|_{BV}),$$
which means $\{W_\Delta \}$ is bounded in $BV(D)$ for any compact set $D \subset Q_T$. Since $BV(D)$ is compactly imbedded into $L^1(D)$, there is a sub-sequences converges in $L^1(D)$ and point-wise almost everywhere in $D$. Next step, we use the diagonal process to construct a sequence converges in $L^1_{loc}(Q_T)$ and point-wise almost everywhere in $Q_T$ to a limit $W$,
$$W(x,t) \in L^{\infty}(Q_T) \cap BV(Q_T).$$

\end{proof}
\textbf{Remark}
It is easy to show that $W$ satisfy the entropy inequality (\ref{eq:entro2}).
Taking $\phi \in C_0^\infty(Q_T),\phi \geq 0$, multiplying the cell inequality (\ref{eq:entro3}) in Lemma 5 by  $\phi dt h$ and summation by parts, we get
\begin{align*}
    &h dt \sum_{i,n} \frac{1}{dt} U^n_i(\phi_i^{n-1}-\phi^n_i) +(a^n_{i-1}\phi^n_{i-1}- a^n_{i+1}\phi^n_{i+1})\frac{1}{2h}U^n_i+F^n_i\frac{1}{2h}(\frac{\phi_{i-1}^n}{r_{i-1}^2}-\frac{\phi_{i+1}^n}{r_{i+1}^2})\\
    &+h dt\sum_{i,n}\frac{1}{h}(\tilde{A}^n_{i+1}\phi^n_{i+1} -\tilde{A}^n_{i}\phi^n_{i} )\frac{1}{h}(U^n_{i+1}-U^n_i)\\
    &-h dt \sum_{i,n}h\frac{\alpha}{2}U^n_i \frac{1}{h^2}(\phi^n_{i+1}-2\phi^n_i+\phi^n_{i-1})+ g^n_i sgn(W^n_i-c)\phi^n_i \leq 0.
\end{align*}

Taking limit $h \to 0$, we have the entropy inequality
\begin{align*}
  \iint_{Q_T} |W-c|\phi_t +(a(x)\phi)_x |W-c|+(\frac{1}{r^2}\phi)_x sgn(W-c)(f(W)-f(c))\\
-(\tilde{A}\phi)_x|W-c|_x  +sgn(W-c) g(W) \phi \ dx dt \geq 0.  
\end{align*}

Finally we have 

\begin{Th}
The sequence $\{A_\Delta,W_\Delta\}$, which is constructed from scheme (\ref{eq136}),
converges in $L^1_{loc}(Q_T)$ and point-wise almost everywhere in $Q_T$ to a BV entropy weak solution of (18)-(20).
\end{Th}

\textbf{Remark} Since $A_i$ can be solved by scheme (\ref{eq:24}), the convergence of numerical for ODE is trivial, so its proof is ignored here.

 \section{WENO schemes for the Yang-Mills equation}

\subsection{A new WENO approximation for the diffusion term}

In this section, we construct a fourth order WENO approximation for $W_{xx}$. Given a uniform grid $\{x_i\}_{i=1}^{N+1} \subset [-5,5]$ with a constant mesh size $h=x_{i+1}-x_i.$ Consider a real value function $W(x)$ defined on interval $[-5,5]$ and denote $W_i=W(x_i)$, we would like to approximate the second order derivative on a 5-point large stencil $S=\{x_{i-2}, x_{i-1}, x_i, x_{i+1}, x_{i+2}\}.$\\

First, let's recall the WENO  scheme of Liu, Shu and Zhang \cite{shu1}.
Consider  a degenerate parabolic equations
\begin{align}\label{eq:100.1}
 u_t=g(u)_{xx}.
\end{align}
One can construct  a conservative finite difference scheme for (\ref{eq:100.1}), written in the form
\begin{align}
    \frac{d}{dt}u_i(t)=\frac{1}{h^2}(\hat{g}_{i+\frac{1}{2}}-\hat{g}_{i-\frac{1}{2}}),
\end{align}
where $u_i(t)$ is the numerical approximation to the point value $u(x_i,t)$ of the solution to (\ref{eq:100.1}), and the numerical flux function is given by
$$\hat{g}_{i+\frac{1}{2}}=\hat{g}(u_{i-r},..,u_{i+s}).$$

The construction of WENO schemes in this section consists of the following steps.

1.Taking  a  big stencil $S=\{ x_{i-r},...,x_{i+r+1}\}$.\\

2. We choose $s$ consecutive small stencils,
$S^{(m)}=\{x_{i-r+m},...,x_{i+r+m+2-s}\}$, $m=0,...,s-1,$ and construct  a series of lower order linear schemes with their numerical fluxes denoted by $\hat{g}^{(m)}_{i+\frac{1}{2}}.$
Here, $s$ can be chosen to be between 2 and 
 and $2r + 1$, corresponding to each small stencil containing $2r + 1$ to 2 points,
 respectively.\\
 
 3. We find the linear weights, namely, constants $d_m$, such that the flux on the big stencil is a linear combination of the fluxes on the small stencils with $d_m$ as the combination coefficients
$$\hat{g}_{i+\frac{1}{2}}=\sum_{m=0}^{s-1} d_m \hat{g}^{(m)}_{i+\frac{1}{2}}.$$

For fourth order  scheme, 
$$d_0=-\frac{1}{12},\qquad d_1=\frac{7}{6}, \qquad d_2=-\frac{1}{12}.$$

4. According to the standard procedure in \cite{Shu1}, we can compute  nonlinear weights $\omega_0,\omega_1,...,\omega_{s-1}$. Finally, we have
$$\hat{g}_{i+\frac{1}{2}}=\sum_{m=0}^{s-1} \omega_m \hat{g}^{(m)}_{i+\frac{1}{2}}.$$

It is exceedingly expensive to compute the six nonlinear weights required to build a 4th-order WENO scheme. Can we construct  a cheaper WENO4 scheme for  equation (\ref{eq:100.1})? In this section, we offer a new technique based on the following simple intuitive:
three points scheme keep the total variation non-increase.

Assume $g'(u) \geq 0$,
consider three points scheme for equation (\ref{eq:100.1}),
\begin{align}
 \frac{1}{dt}(u_i^{n+1}-u^n_i)=\frac{1}{h^2}(g_{i+1}^n-2g_i^n+g_{i-1}^n) ,  
\end{align}
which can be written  as 
\begin{align}\label{eq:47}
    u^{n+1}_i =u^n_{i}-\frac{dt}{h^2}g'(\theta^n_{i-\frac{1}{2}})(u^n_i-u^n_{i-1})+\frac{dt}{h^2}g'(\theta^n_{i+\frac{1}{2}})(u^n_{i+1}-u^n_{i}),
\end{align}
where $\theta^n_{i-\frac{1}{2}}$ is between $u^n_{i-1}$ and $u^n_{i}.$ Equation (\ref{eq:47}) satisfies Harten's Lemma, so $$TV(u^{n+1})\leq TV(u^{n}),$$
which means the three points scheme has no oscillation. So, we can divided the big stencil $S$  into two sub-stencils $\{S^0,S^1\}$, 
where $S^0=\{x_{i-1},x_i, x_{i+1}\},S^1=\{ x_{i-2}, x_i, x_{i+2}\}.$
The fourth-order approximation $\hat{g}_{xx,i}=g_{xx}(x_i)+O(h^4)$ is built through the convex combinations of  $\hat{g}^{(k)}_{xx,i}$, defined in each one of the stencils $S^k$:
\begin{equation*}
\hat{g}_{xx,i}=\omega_0 \hat{g}^{(0)}_{xx,i}+\omega_1 \hat{g}^{(1)}_{xx,i}.
\end{equation*}
If there is shock in big stencil $S$,one can use more weights of stencil $S^0$ and less weights of $S^1$. In this way, we can avoid the oscillation.

For EYM equation, we need to approximate $W_{xx}$ using WENO scheme. We described this process as following.

The fourth-order approximation $\hat{W}_{xx,i}=W_{xx}(x_i)+O(h^4)$ is built through the convex combinations of  $\hat{W}^{(k)}_{xx,i}$, defined in each one of the stencils $S^k$:
\begin{equation}
\hat{W}_{xx,i}=\omega_0 \hat{W}^{(0)}_{xx,i}+\omega_1 \hat{W}^{(1)}_{xx,i},
\end{equation}
 where
 \begin{align}
 \hat{W}^{(0)}_{xx,i} &=\frac{1}{h^2}(W_{i+1}-2W_i+W_{i-1}), \\
 \hat{W}^{(1)}_{xx,i}&= \frac{1}{4h^2}(W_{i+2}-2W_i+W_{i-2}) .
 \end{align} 
 
 The $W_{xx}$ can be approximated in the big stencil 
 $$\hat{W}_{xx,i}= \frac{1}{12 h^2}(-W_{i-2}+16 W_{i-1}-30 W_i+16 W_{i+1}-W_{i+2}), $$
 and $$\hat{W}_{xx,i}=d_0 W_{xx,i}^{(0)}+ d_1 W_{xx,i}^{(1)},$$
 where $d_k$ are linear weights 
 $$d_0=\frac{4}{3},\qquad d_1=-\frac{1}{3}.$$ 
 To handle this negative weights, we  consider the following standards procedure 
 
  \begin{align}
 \tilde{\gamma}^+_k &=\frac{1}{2}(d_k+\theta |d_k|),k=0,1,\theta=2,\\
 \tilde{\gamma}^-_k &= \tilde{\gamma}^+_k -d_k,\\
 \sigma^+&=\tilde{\gamma}^+_0+\tilde{\gamma}^+_1=\frac{13}{6}, \\
  \sigma^-&=\tilde{\gamma}^-_0+\tilde{\gamma}^-_1 =\frac{7}{6}.
\end{align}

 Then we obtain the following 
 \begin{align}
  \gamma^+_0 &=\frac{ \tilde{\gamma}^+_0}{\sigma^+}=\frac{12}{13},\\
  \gamma^+_1 &=\frac{ \tilde{\gamma}^+_1}{\sigma^+}=\frac{1}{13},\\
   \gamma^-_0 &=\frac{ \tilde{\gamma}^-_0}{\sigma^-}=\frac{4}{7},\\
  \gamma^-_1 &=\frac{ \tilde{\gamma}^-_1}{\sigma^-}=\frac{3}{7}.
 \end{align}
 
 The smoothness indicators $\beta_0,\beta_1$ are computed as
 $$ \beta_{k}=\int_{x_{i-1}}^{x_{i+1}} h (\frac{d  p_k}{d x})^2 dx + \int_{x_{i-1}}^{x_{i+1}} h^3 (\frac{d^2  p_k}{d x^2})^2 dx, \quad k=0, 1, $$
 \begin{align}
 \beta_0 &=\frac{1}{2}( W_{i+1}-W_{i-1})^2+\frac{8}{3}(W_{i+1} -2W_i+W_{i-1} )^2,\\
 \beta_1 &=\frac{1}{8}(W_{i+2}-W_{i-2})^2+\frac{1}{6}(W_{i+2}-2W_i+W_{i-2})^2.
 \end{align}

 We  obtained the  nonlinear weights  by
 \begin{align} \label{eq28}
 \alpha^{\pm}_k &=\frac{\gamma^{\pm}_k}{(\epsilon +\beta_k)^2},\\
 \omega^{\pm}_k &=\frac{\alpha^\pm_k}{\alpha^\pm_0 + \alpha^\pm_1},\\
 \omega_k&=\sigma^+ \omega^+_r -\sigma^- \omega^-_k, k=0,1
 \end{align}
 where $\epsilon$ is used to avoid the division by zero in the denominator. We take $\epsilon=10^{-10}.$
 
 Then we have
 
 \begin{equation}
 W_{xx}=\omega_0 W^{(0)}_{xx} + \omega_1 W^{(1)}_{xx}.\label{eq:w_xx}
 \end{equation}
 
% \subsection{Accuracy analysis of the $W_{xx}$}
 We drive a  sufficient condition for fourth order convergence of  Eq.~(\ref{eq:w_xx}).
 Adding and subtracting $\displaystyle \sum_{k=0}^1  d_k \hat{W}^{(k)}_{xx} $ from Eq. (\ref{eq:w_xx})  give
 $$\hat{W}_{xx} =\sum_{k=0}^1  d_k \hat{W}^{(k)}_{xx} + \sum_{k=0}^1  (\omega_k -d_k) \hat{W}^{(k)}_{xx} ,$$
 where the first term on the right hand  side produces the 4th order accurate. The second term must be at least $O(h^4)$ in order  for $\hat{W}_{xx}$ to be approximated at  4th order. Noting that the $W^{(k)}_{xx}$ are 2nd order approximations of $\hat{W}_{xx}(x_i)$, we have
 \begin{align}
 \sum_{i=0}^1 (\omega_k-d_k) \hat{W}^{(k)}_{xx} &= \sum_{i=0}^1 (\omega_k-d_k)  (W_{xx}+O(h^2)) \nonumber\\
 &= \sum_{i=0}^1 (\omega_k-d_k) W_{xx} (x_i) +  \sum_{i=0}^1 (\omega_k-d_k) (O(h^2)),
 \end{align}
 where the first term on the right hand side vanishes due to the normalization of the weights. Thus, it is sufficient to require 
 $$\omega_k=d_k +O(h^2).$$
 
 By  Taylor expansion, we have 
 \begin{align}
 \beta_0 &=\frac{8}{3}(W_{j+1}-2W_j+W_{j-1})^2+\frac{1}{2}(W_{j+1}-W_{j-1})^2\nonumber\\
 &=\frac{8}{3}\left( h^2 W_{xx}(x_j) \right)^2+\frac{1}{2}\left(2h W_x(x_j)+\frac{1}{3}h^3 W_{xxx}(x_j)\right)^2+O(h^6),\\
 \beta_1&=\frac{1}{6}(W_{j+2}-2W_j+W_{j-2})^2+\frac{1}{8}(W_{j+2}-W_{j-2})^2\nonumber\\
              &=\frac{1}{6}\left( (2h)^2W_{xx}(x_j) \right)^2+\frac{1}{8}\left( 4h W_{x}(x_j) +\frac{1}{3}(2h)^3W_{xxx}(x_j) \right)^2+O(h^6).
 \end{align}
 
 If $\displaystyle W_x(x_j) \neq 0$,  $\beta_0=2h^2\left(W_x(x_j)\right)^2[1+O(h^2)]$, $\beta_1=2h^2\left(W_x(x_j)\right)^2[1+O(h^2)]$. 
 
 If $\displaystyle W_x(x_j)=0,W_{xx}(x_j) \neq 0$, $\beta_0=\frac{8}{3} h^4 \left(W_{xx} (x_j)\right)^2[1+O(h^2)]$,$\beta_1=\frac{4}{6} h^4 \left(W_{xx} (x_j)\right)^2[1+O(h^2)]$.
 
 Then we have $\beta_k=D\left(1+O(h^2) \right), k=0,1$,  where $D$ is a constant   independent of the $k$.
 By the  Taylor expansion 
 \begin{align}
 \frac{\gamma^\pm_k }{(\varepsilon +\beta_k)^2} &=\frac{\gamma^\pm_k}{D^2 ( 1+ O(h^2) )^2}, \nonumber\\
                                                                                &=\frac{\gamma^\pm_k}{D^2}(1+O(h^2) ),
    \end{align}
%By  the definition of non[]
then
\begin{align}
\gamma_k^{\pm} &=\omega^\pm_k ( \sum_{\ell =0}^1 \frac{ \gamma^\pm_\ell}{( \varepsilon +\beta_\ell)^2})(\varepsilon+\beta_k)^2\nonumber\\
                                &=\omega^\pm_k(\frac{1}{D^2}\left(1+O(h^2) )  \right)\left(D(1+O(h^2)) \right)^2 \nonumber\\
                                &=\omega^\pm_k+ O(h^2).
\end{align}
By the definition of $\omega_k$
\begin{align}
\omega_k &=\sigma^+ \omega_k^+ - \sigma^- \omega_k^- \nonumber\\
                   &=\sigma^+(\gamma^+_k+O(h^2))  -\sigma^-(\gamma^-_k+O(h^2)) \nonumber\\
                   &=d_k+O(h^2).
\end{align}
 To achieve fourth order accuracy in critical points, we fix the nonlinear weights $\omega_k,k=0,1$ by a mapping function \cite{map15}
 
 \begin{equation}\label{eq:78mp}
     g_k(\omega)=\frac{\omega(d_k+d_k^2-3d_k\omega+\omega^2)}{d_k^2+\omega(1-2d_k)}.
 \end{equation}
 The mapped nonlinear weights are given by
 \begin{align*}
     \alpha_k&=g_k(\omega_k),k=0,1,\\
     \omega_k^{new}&=\frac{\alpha_k}{\alpha_0+\alpha_1}.
 \end{align*}
 Then, we replace the original nonlinear weights in (\ref{eq:w_xx}) by $\omega_k^{new}$.
 This method worked well.

On the other hand,we propose a simple modified limiting procedure:
$$ \beta_k =\left\{
\begin{aligned}
&0,           &R(\beta)  <\xi\\
&\beta_k,& otherwise
\end{aligned}
\right.
$$

where

$$R(\beta)=\max\limits_{0\leq k\leq 1} \beta_k,$$
and $\xi$ can be chosen suitably. 
The basic idea behind this method is that
in the smooth region, we just need the nonlinear weights to be the linear weights $d_k$.Then there is no accuracy loss phenomena in the critical point.

 %===================================================
 \subsection{WENO5 scheme for the $W_x$ and convection term }
 
We give  the  left bias fifth order finite difference  WENO  approximate of the first  derivative  $W_x$ at the grid point $x_j$:
\begin{align}
 W_{x,j}^{-} =\frac{1}{h}(\hat{W}_{j+\frac{1}{2}}-\hat{W}_{j-\frac{1}{2}}).
\end{align}
The  numerical flux $\hat{W}_j$ is given by
\begin{align}\label{weigh3}
	\hat{W}_{j+\frac{1}{2}}=\omega_1 \hat{W}^{(1)}_{j+\frac{1}{2}}+\omega_2 \hat{W}^{(2)}_{j+\frac{1}{2}}+ \omega_3 \hat{W}^{(3)}_{j+\frac{1}{2}},
\end{align}
where $\hat{W}^{(i)}_{j+\frac{1}{2}},i=1,2,3$, are three third order fluxes on the three difference small  stencils given by 
\begin{align}
	\hat{W}^{(1)}_{j+\frac{1}{2}} &=\frac{1}{3}W_{j-2}-\frac{7}{6}W_{j-1}+\frac{11}{6}W_j, \\
	\hat{W}^{(2)}_{j+\frac{1}{2}} &=-\frac{1}{6}W_{j-1}+\frac{5}{6}W_{j}+\frac{1}{3}W_{j+1}, \\
	\hat{W}^{(3)}_{j+\frac{1}{2}} &=\frac{1}{3}W_{j}+\frac{5}{6}W_{j+1}-\frac{1}{6}W_{j+2} .
\end{align}
The nonlinear weights $\omega_i$ are given by 
\begin{align}
	\omega_i &=\frac{\tilde{\omega}_i}{\sum_{k=1}^{3} \tilde{\omega }_k},\\
	\tilde{\omega}_k &=\frac{\gamma_k}{ ( \varepsilon+\beta_k )^2},
\end{align}
with the linear weights $\gamma_k$ given by 
$$\gamma_1 =\frac{1}{10},\qquad \gamma_2=\frac{6}{10},\qquad\gamma_3=\frac{3}{10}.$$
The smoothness indicators $\beta_k$ are given by 
\begin{align}
	\beta_1 &=\frac{13}{12}(W_{j-2}-2W_{j-1}+W_j )^2   \nonumber \\
	               &~~~+\frac{1}{4}(W_{j-2}-4W_{j-1}+3W_j )^2, \\
	\beta_2 &=\frac{13}{12}(W_{j-1}-2W_j+W_{j+1} )^2  \nonumber \\
	                 &~~~+\frac{1}{4}(W_{j-1}-W_{j+1})^2, \\
	\beta_3 &=\frac{13}{12}(W_{j}-2W_{j+1}+W_{j+2} )^2  \nonumber\\
	                 &~~~+\frac{1}{4}(3W_j-4W_{j+1}+W_{j+2})^2.
\end{align}

The right bias fifth order finite difference  WENO  approximate $W^{+}_x$ is mirror symmetric to that for $W^-_x$.
% The  convection term $-c W_x$ can  be approximated by upwind  WENO scheme
% \begin{align}
% 	-\frac{1}{2}(c+\max|c|)W^{n-}_{x,j}-\frac{1}{2}(c-\max|c|)W^{n+}_{x,j}.
% \end{align}

To achieve fifth order accuracy in the critical point, one can fix the nonlinear weights $\omega_k$ by a mapping described in Eq. (\ref{eq:78mp}). For more details, one can see \cite{map15}. Another new idea is given by R. Borges \cite{Borg}, which is called WENO-Z scheme. The novel method is to use the big stencil to construct a new smoothness indicator of higher order than the classical smoothness indicators $\beta_k$. The new indicator is denoted as $\tau_5$
 $$\tau_5=|\beta_3-\beta_1|.$$
We define the new smoothness indicators $$\beta_k^{z}=( \frac{\beta_k+\varepsilon}{\beta_k+\tau_5+\varepsilon}),   \quad k=1,2,3,$$
and the new WENO weights ${\omega}^z_k$ as
\begin{align}\label{eq:Borg}
    \omega^z_k &=\frac{\alpha^z_k}{\sum_{i=1}^3\alpha^z_i}\\
    \alpha^z_i&=d_i(1+(\frac{\tau_5}{\beta_i+\varepsilon})^q), \quad i=1,2,3,
\end{align}
where $\varepsilon=10^{-10},q=2.$

Next, we consider the WENO5 approximations for $f(W)_x.$ Since $f'(W)=(1-W^2)^2 \geq 0$, then $f(W)_x$  at $x_j$ can be represented  as \cite{shu2}
\begin{equation}\label{eq:flux}
 \frac{1}{h}(\hat{f}_{j+\frac{1}{2}}-\hat{f}_{j-\frac{1}{2}}).   
\end{equation}

The numerical flux $\hat{f}_{j+\frac{1}{2}}$ can be reconstructed by the point value $f(W_j)$ as the  following procedure

\begin{align*}
	\hat{f}_{j+\frac{1}{2}}=\omega_1 \hat{f}^{(1)}_{j+\frac{1}{2}}+\omega_2 \hat{f}^{(2)}_{j+\frac{1}{2}}+ \omega_3 \hat{f}^{(3)}_{j+\frac{1}{2}},
\end{align*}
where $\hat{f}^{(i)}_{j+\frac{1}{2}},i=1,2,3$, are three third order fluxes on the three difference small  stencils given by 
\begin{align*}
	\hat{f}^{(1)}_{j+\frac{1}{2}} &=\frac{1}{3}f_{j-2}-\frac{7}{6}f_{j-1}+\frac{11}{6}f_j, \\
	\hat{f}^{(2)}_{j+\frac{1}{2}} &=-\frac{1}{6}f_{j-1}+\frac{5}{6}f_{j}+\frac{1}{3}f_{j+1}, \\
	\hat{f}^{(3)}_{j+\frac{1}{2}} &=\frac{1}{3}f_{j}+\frac{5}{6}f_{j+1}-\frac{1}{6}f_{j+2} .
\end{align*}
The nonlinear weights $\omega_i$ are given by 
\begin{align*}
	\omega_i &=\frac{\tilde{\omega}_i}{\sum_{k=1}^{3} \tilde{\omega }_k},\\
	\tilde{\omega}_k &=\frac{\gamma_k}{ ( \varepsilon+\beta_k )^2},
\end{align*}
with the linear weights $\gamma_k$ given by 
$$\gamma_1 =\frac{1}{10},\qquad\gamma_2=\frac{6}{10},\qquad\gamma_3=\frac{3}{10}.$$
The smoothness indicators $\beta_k$ are given by 
\begin{align*}
	\beta_1 &=\frac{13}{12}(f_{j-2}-2f_{j-1}+f_j )^2   \nonumber \\
	               &~~~+\frac{1}{4}(f_{j-2}-4f_{j-1}+3f_j )^2, \\
	\beta_2 &=\frac{13}{12}(f_{j-1}-2f_j+f_{j+1} )^2  \nonumber \\
	                 &~~~+\frac{1}{4}(f_{j-1}-f_{j+1})^2, \\
	\beta_3 &=\frac{13}{12}(f_{j}-2f_{j+1}+f_{j+2} )^2  \nonumber\\
	                 &~~~+\frac{1}{4}(3f_j-4f_{j+1}+f_{j+2})^2.
\end{align*}
We can fix this nonlinear weights using (\ref{eq:Borg}).

 %===================================================

\subsection{Explicit/Implicit WENO scheme for the Yang-Mills equation}
 
In this section we construct the  fourth order WENO scheme for the convection-diffusion equation

 $$W_t+a(x)W_x+\frac{1}{r^2}f(W)_x=\tilde{A}W_{xx}+(1-W^2)W.$$
 Define 
 $$a_i^+=\frac{1}{2}(a_i+ |a_i| ),a_i^-=\frac{1}{2}(a_i- |a_i| ).$$
  We have the following  explicit WENO scheme
  
 $$ \tilde{A}^n_i\hat{W}^n_{xx,i}+(1-(W_i^n)^2)W_i^n= \frac{1}{dt}( W_i^{n+1}-W_i^n ) +a^+_i\hat{W}_i^{n-}  + a^-_i\hat{W}_i^{n+}+\frac{1}{r_i^2h}(\hat{f}^n_{i+\frac{1}{2}}-\hat{f}^n_{i-\frac{1}{2}}),$$
 where $\hat{f}^n_{i\pm\frac{1}{2}}$ is constructed in (\ref{eq:flux}).
 
 The CFL condition is  given by $$dt < \frac{0.3 h^2}{2\max_i A_i+\max_i|c_i|h},$$
 where 
$$c=2A-1+\frac{1}{r^2}(1-W^2)^2.$$

 % \section{Implicit  WENO scheme for Yang-Mills  equation}
  To accelerate  decay, we design an implicit scheme for EYM equations as following 
  \begin{equation}\label{eq:weno_w} 
      \tilde{A}^n_i\hat{W}^{n+1}_{xx,i}+(1-(W_i^n)^2)W_i^n=
    \frac{1}{dt}( W_i^{n+1}-W_i^n ) +a^+_i\hat{W}_i^{n-}  + a^-_i\hat{W}_i^{n+}+\frac{1}{r_i^2h}(\hat{f}^n_{i+\frac{1}{2}}-\hat{f}^n_{i-\frac{1}{2}}), 
  \end{equation}
   where $\hat{W}^{n+1}_{xx,i}$ can be approximated by
   $$\hat{W}^{n+1}_{xx,i}=\omega_0 \frac{1}{h^2}(W_{i+1}^{n+1} -2W_i^{n+1}+W^{n+1}_{i-1}) + \omega_1 \frac{1}{4h^2}(W_{i+2}^{n+1} -2W_i^{n+1}+W^{n+1}_{i-2}),$$
   and the nonlinear weights $\omega_k, k=0,1$ are calculated by $W^n$ as (\ref{eq28}).

\section{WENO schemes for the Einstein constraint equation}

\subsection{WENO type Adams solver }
First step, we consider a simple ODE problem 
 \begin{align}
 y_x&=f(x),\\
 y(x_0)&=y_0,
 \end{align}
 where $f(x)$ is discontinued or very sharp  at somewhere. The high order integral could cause numerical oscillation  near the discontinued point. 
Integral the equation in interval $I_{i+1/2}=[x_i,x_{i+1}]$, then
$$y_{i+1}-y_i=\int_{x_i}^{x_{i+1}} f(x) dx,$$
where  $\displaystyle \int_{x_i}^{x_{i+1}} f(x) dx$ can be approximated by the  WENO integral as following.

 We  chose two  sub stencils $S^0=\{ x_{i-2}, x_{i-1}, x_i\},S^1=\{ x_{i-1}, x_{i}, x_{i+1}\}.$
 There is a unique polynomial $p_r(x)$ of degree at most $2$ which interpolates $ f(x)$ at the nodes in $S^r, r=0,1.$ We denote the integral of $p_r(x)$ on $I_{i+\frac{1}{2}}$ by $J^{(r)} ,r=0,1$

 $$ J^{(0)}=\frac{h}{12}(23 f_i-16 f_{i-1}+5f_{i-2}), J^{(1)}=\frac{h}{12}(8f_i-f_{i-1}+5f_{i+1}). $$

In  the  large stencil  $S=\{  x_{i-2}, x_{i-1}, x_i, x_{i+1}\}, $ the  integral of $ \displaystyle \int_{x_i}^{x_{i+1}}  f(x) dx$ can be approximated with the linear  coefficients $$J=\frac{h}{24}(19 f_i-5 f_{i-1}+f_{i-2} +9 f_{i+1}) .$$

  The WENO integration would take a convex combination of $J^{(r)}$ defined above  as a new approximation to the integral $\displaystyle  \int_{I_{i+\frac{1}{2}}} f(x) dx $
  
  $$ J=\omega_0 J^{(0)}+\omega_1J^{(1)}.$$
  
  We ask the nonlinear $\omega_r \geq 0$ and $\omega_0+\omega_1=1$ for stability and consistency.
  
  We know that  for smooth $g(x)$, then
  $$J=d_0 J^{(0)}+d_1 J^{{1}} =\int_{I_{i+\frac{1}{2}}} f(x) dx+ O(h^{2k-1}),$$
  Where $d_0=\frac{1}{10},d_1=\frac{9}{10}.$

   In the smooth case, we hope to have $\omega_r=d_r+O(h^{k-1}),r=0,1$ so that $5th$ order accuracy can be achieved for the integral. When the function $g$ is discontinued at  one stencil, we ask the corresponding weights $\omega_r$ to essentially 0 to  avoid  oscillation.    We can construct the  nonlinear weights as following:
  
  First, we  construct the smoothness indicators in every small stencil $S_i^r, r=0,1$, 
  $$\beta_{r}=\int_{x_{i-\frac{1}{2}}}^{x_{i+\frac{1}{2}}} h (\frac{d  p_r}{d x})^2  +  h^3 (\frac{d^2  p_r}{d x^2})^2 dx  .$$
  Then we have 
  \begin{align}
  \beta_0 &=\frac{1}{4}( 3 f_i-4 f_{i-1}+f_{i-2})^2+\frac{13}{12}(f_i-2f_{i-1}+f_{i-2})^2, \\
  \beta_1 &=\frac{1}{4}(f_{i+1}-f_{i-1})^2+\frac{13}{12}(f_{i+1}-2f_{i}+f_{i-1})^2,
  \end{align}
  or
 $$\beta_{r}=\int_{x_i}^{x_{i+1}} h (\frac{d  p_r}{d x})^2  +  h^3 (\frac{d^2  p_r}{d x^2})^2 dx  .$$
  Then we have 
  \begin{align}
  \beta_0 &=( 2 f_i-3f_{i-1}+f_{i-2})^2+\frac{13}{12}(f_i-2f_{i-1}+f_{i-2})^2, \\
  \beta_1 &=(f_{i+1}-f_{i})^2+\frac{13}{12}(f_{i+1}-2f_{i}+f_{i-1})^2.
  \end{align}

  Second step, we construct $\alpha_r, r=0,1.$
  $$\alpha_r=\frac{d_r}{(\varepsilon +\beta_r)^2}, r=0,1.$$
  
  Finally, we get  the nonlinear weights for central  WENO  integral 
  $$\omega_r=\frac{\alpha_r}{\alpha_0+\alpha_1},r=0,1.$$
   \begin{align}
   y_{i+1} -y_i  &=\omega_0 \frac{h}{12} (23f_i  -16 f_{i-1}  + 5 f_{i-2}  )\nonumber\\
                                                            &+\omega_1 \frac{h}{12} (8f_i  - f_{i-1}  + 5 f_{i+1} ).
   \end{align}
   
   In the smooth region $\omega_k=d_k$,     which is the familiar Adams-Moulton 4  scheme:
   $$ y_{i+1}-y_i =\frac{h}{24}(19 f_i+9 f_{i+1}-5f_{i-1}+f_{i-2}).$$
   At the left  boundary, we use the following scheme
   \begin{align}
   y_{i+1}-y_{i}=\frac{h}{24}(19 f_{i+1} +9f_{i} -5 f_{i+2} +f_{i+3}),i=1.
  \end{align}
   
 \subsection{Three sub stencils WENO integration }   
 In this section, we design   WENO integration on three sub stencils.
 %We design a three sub-stencils method in this subsection.
  Consider  sub stencils $S^0=\{ x_{i-2},x_{i-1}\},\ S^1=\{x_{i-1},x_i\},\ S^2=\{x_i, x_{i+1}\}.$
 Denote $p_k(x),k=0,1,2$ be the first order Interpolation polynomials of $f(x)$ at each sub stencils and $J^k, k=0,1,2$ are the  integral of $p_k(x)$ on the interval 
 $[x_i, x_{i+1}]$,
 \begin{align}
 J^0&=\int_{x_i}^{x_{i+1}} p_0(x) dx =\frac{h}{2}(5f_{i-1}-3f_{i-2}), \\
 J^1&=\int_{x_i}^{x_{i+1}} p_1(x) dx =\frac{h}{2}(-f_{i-1}+3f_{i}), \\
 J^2&=\int_{x_i}^{x_{i+1}} p_2(x) dx =\frac{h}{2}(f_{i}+3f_{i+1}) .
 \end{align}
 
 Integrate  $f(x)$ on the large stencils, we have 
 \begin{align*}
 J&=\frac{h}{24}(f_{i-2}-5f_{i-1}+19f_i+9f_{i+1})\\
  &=d_0 J^0+d_1 J^1+d_2 J^2,
 \end{align*}
 where the linear weights $d_k, k=0,1,2$ are given by
 \begin{align*}
   d_0&=-\frac{1}{36},\\
   d_1&=\frac{10}{36},\\
   d_2&=\frac{27}{36}.
 \end{align*}
 To handle this negative weights, we  consider the following standards procedure. 
 
We define 
\begin{align*}
\tilde{\gamma}^+_k &=\frac{1}{2}(d_k +\theta |d_k|),\theta=3,k=0,1,2\\
\tilde{\gamma}^-_k  &=\tilde{\gamma}^+_k-d_k,
\end{align*}
and 
\begin{align*}
\tilde{\gamma}_0^+ &=\frac{1}{36},\qquad \tilde{\gamma}_0^- =\frac{2}{36} ,\\
\tilde{\gamma}_1^+ &=\frac{20}{36}, \qquad\tilde{\gamma}_1^- =\frac{10}{36}, \\
\tilde{\gamma}_2^+ &=\frac{54}{36}, \qquad\tilde{\gamma}_2^- =\frac{27}{36}. \\
\end{align*}

Define $\sigma^\pm$ as follows
\begin{align*}
\sigma^+ &=\sum_{\ell=0}^2 \tilde\gamma_\ell^+=\frac{75}{36},\\
\sigma^- &=\sum_{\ell=0}^2  \tilde\gamma_\ell^-=\frac{39}{36}.
\end{align*}

Define $\displaystyle  \gamma_k^\pm =\frac{\tilde{\gamma}_k^\pm}{\sigma^\pm},$
 then
\begin{align*}
\gamma^+_0 &=\frac{1}{75} ,\qquad\gamma^-_0=\frac{2}{39}, \\
\gamma^+_1 &=\frac{20}{75} ,\qquad\gamma^-_1=\frac{10}{39}, \\
\gamma^+_2 &=\frac{54}{75} ,\qquad\gamma^-_1=\frac{27}{39}. \\
\end{align*}

The smoothness indicators $\beta_k ,k=0,1,2$ are given by
\begin{align*}
\beta_0&=(f_{i-1}-f_{i-2})^2 ,\\
\beta_1&=(f_{i}-f_{i-1})^2 ,\\
\beta_2&=(f_{i+1}-f_{i})^2 .
\end{align*}

The nonlinear weights are computed by
\begin{align*}
\alpha_k^{\pm} &=\frac{\gamma^\pm_k}{(\varepsilon +\beta_k)^2}, \\
\omega^\pm_k  &=\frac{\alpha_k^\pm}{\sum_{j=0}^2 \alpha_j^\pm}.
\end{align*}
Then
\begin{equation}\label{eq123}
\omega_k=\sigma^+ \omega_k^+-\sigma^-\omega_k^-, \quad k=0,1,2.
\end{equation}
Finally, we have  a WENO  solver for the ODE $y_x=f(x)$:
\begin{align}\label{eq125}
y_{i+1}-y_i&=\omega_0 J^0+\omega_1 J^1+\omega_2 J^2\nonumber\\
                     &=\omega_0 \frac{h}{2}(5f_{i-1}-3f_{i-2})+\omega_1 \frac{h}{2}(-f_{i-1}+3f_i) +\omega_2 \frac{h}{2}(f_i+f_{i+1}).
\end{align}

\subsection{WENO type Adams solver for constraint equations}

We define the auxiliary function $\Box, q, S$ as  follows
 \begin{align*}
 \Box&=2W_r^2,\\
 S&=1-\frac{1}{r^2}(1-W^2)^2,\\
 q_x &=1+\Box.
 \end{align*}

In this section, we consider the Einstein constraint equation  $$  (Ae^q)_x =S e^q.$$
Integrate in the interval $I_{i+\frac{1}{2}}=[x_i,x_{i+1}]$, $$ A_{i+1} e^{q_{i+1}}- A_{i} e^{q_{i}} = \int_{x_i}^{x_{i+1}} S e^q dx .$$

$\displaystyle \int_{x_i}^{x_{i+1}}  S e^q dx$ can be approximate by the  WENO  described   above.

 We  chose two  sub stencils $S^0=\{ x_{i-2}, x_{i-1}, x_i\},S^1=\{ x_{i-1}, x_{i}, x_{i+1}\}.$
 There is a unique polynomial $p_r(x)$ of degree at most $2$ which interpolates $ g:=S e^q$ at the nodes in $S^r, r=0,1.$ We denote the integral of $p_r(x)$ on $I_{i+\frac{1}{2}}$ by $J^{(r)} ,r=0,1$,

 $$ J^{(0)}=\frac{h}{12}(23 g_i-16 g_{i-1}+5g_{i-2}), J^{(1)}=\frac{h}{12}(8g_i-g_{i-1}+5g_{i+1}). $$
 
In  the  large stencil  $S=\{  x_{i-2}, x_{i-1}, x_i, x_{i+1}\},$ the  integral of $ \displaystyle \int_{x_i}^{x_{i+1}}  g dx$ can be approximated with the linear  coefficients $$J=\frac{h}{24}(19 g_i-5 g_{i-1}+g_{i-2} +9 g_{i+1}) .$$

  The WENO integration would take a convex combination of $J^{(r)}$ defined above  as a new approximation to the integral $\displaystyle  \int_{I_{i+\frac{1}{2}}} g(u,x) dx $
  
  $$ J=\omega_0 J^{(0)}+\omega_1J^{(1)}.$$
  
  We ask the nonlinear $\omega_r \geq 0$ and $\omega_0+\omega_1=1$ for stability and consistency.
  
  We know that  for smooth $g(x)$, then
  $$J=d_0 J^{(0)}+d_1 J^{{1}} =\int_{I_{i+\frac{1}{2}}} g(x) dx+ O(h^{2k-1}),$$
  
  where $d_0=\frac{1}{10},d_1=\frac{9}{10}.$

   In the smooth case, we hope to have $\omega_r=d_r+O(h^{k-1}),r=0,1$ so that $5th$ order accuracy can be achieved for the integral. When the function $g$ is discontinued at  one stencil, we ask the corresponding weights $\omega_r$ to essentially 0 to  avoid  oscillation.    We can construct the  nonlinear weights as following.
  
  First ,we  construct the smoothness indicators in every small stencil $S_i^r, r=0,1,$
  $$\beta^{r}=\int_{x_{i-\frac{1}{2}}}^{x_{i+\frac{1}{2}}} h (\frac{d  p_r}{d x})^2  +  h^3 (\frac{d^2  p_r}{d x^2})^2 dx  .$$
  Then we have 
  \begin{align}
  \beta_0 &=\frac{1}{4}( 3 g_i-4 g_{i-1}+g_{i-2})^2+\frac{13}{12}(g_i-2g_{i-1}+g_{i-2})^2, \\
  \beta_1 &=\frac{1}{4}(g_{i+1}-g_{i-1})^2+\frac{13}{12}(g_{i+1}-2g_{i}+g_{i-1})^2.
  \end{align}
  
  Second step, we construct $\alpha_r, r=0,1$,
  $$\alpha_r=\frac{d_r}{(\varepsilon +\beta_r)^2}, \quad r=0,1.$$
  
  Finally, we get  the nonlinear weights for central  WENO  integral 
  $$\omega_r=\frac{\alpha_r}{\alpha_0+\alpha_1},r=0,1,$$

   %which is the familiar Adams-Moulton 4  scheme. 
   
   \begin{align}
   A_{i+1} -A_i e^{q_i-q_{i+1}} &=\omega_0 \frac{h}{12} (23S_i e^{q_i-q_{i+1}} -16 S_{i-1} e^{q_{i-1}-q_{i+1}} + 5 S_{i-2} e^{q_{i-2}-q_{i+1}} )\nonumber\\
                                                            &+\omega_1 \frac{h}{12} (8S_i e^{q_i-q_{i+1}} - S_{i-1} e^{q_{i-1}-q_{i+1}} + 5 S_{i+1} ),
   \end{align} 
   where the $\displaystyle q_j-q_i=\int_{x_i} ^{x_j} 1+\Box \  dx $ and the integral can be approximated by the WENO procedure  described above. 
   
   \textbf{Remark} To approximate the $\Box=\frac{2}{r^2}W_x^2$, we only use linear scheme. However,  in order to avoid possible oscillations, we follow a simple "min-mod" principle.
   We define the numerical approximation of $W_x^2$ in $x_i$ is
   $\hat{W}_{x,i}^2$,which can be given by
   $$\hat{W}_{x,i}^2=\min \left((W_{x,i}^-)^2,(W_{x,i}^+)^2, (W^c_{x,i})^2  \right),$$
   where $W^c_{x,i}=\frac{1}{2}(W_{x,i}^- +W_{x,i}^+),$
    $W^-_{x,i}$ is the left bias fifths  order linear approximation of $W_x$ and $W^+_{x,i}$ is the right bias fifths  order approximation.
   \subsection{WENO-Admas schemes in three sub stencils }
   Define 
   \begin{align}\label{eq:138}
   g_{i-2}&=S_{i-2} e^{q_{i-2}-q_{i+1}},\\
   g_{i-1}&=S_{i-1}e^{q_{i-1}-q_{i+1}},\\
   g_i &=S_i e^{q_i-q_{i+1}} ,\\
   g_{i+1}&=S_{i+1}.
   \end{align}
   Then we have the  three sub stencils WENO type Admas schemes
   \begin{equation}\label{eq:139}
    A_{i+1}-A_i e^{q_i-q_{i+1}} =\omega_0 \frac{h}{2}(5g_{i-1}-3g_{i-2})+\omega_1 \frac{h}{2}(-g_{i-1}+3g_i) +\omega_2 \frac{h}{2}(g_i+g_{i+1}),
   \end{equation}
   where the nonlinear weights $\omega_k$  are computed as (\ref{eq123}) and $\displaystyle  q_j-q_i=\int_{x_i}^{x_j}  1+\Box \  dx $ are calculated  by the method described  in  (\ref{eq125}).
   
   At the left  boundary, we use the following  linear scheme
   \begin{align}
   A_{i+1}-A_i e^{q_i-q_{i+1}}=\frac{h}{24}(19 S_{i+1} +9S_{i}e^{q_i-q_{i+1}} -5 S_{i+2}e^{q_{i+2}-q_{i+1}} +S_{i+3}e^{q_{i+3}-q_{i+1}}),i=1.
  \end{align}
  We combine  (\ref{eq:weno_w}) and (\ref{eq:139}) together
  \begin{align}
       A_{i+1}-A_i e^{q_i-q_{i+1}} =\omega_0 \frac{h}{2}(5g_{i-1}-3g_{i-2})+\omega_1 \frac{h}{2}(-g_{i-1}+3g_i) +\omega_2 \frac{h}{2}(g_i+g_{i+1}),\label{eq:141}\\
          \tilde{A}_i\hat{W}^{n+1}_{xx,i}+(1-(W_i^n)^2)W_i^n=
    \frac{1}{dt}( W_i^{n+1}-W_i^n ) +a^+_i\hat{W}_i^{n-}  + a^-_i\hat{W}_i^{n+}+\frac{1}{r_i^2h}(\hat{f}^n_{i+\frac{1}{2}}-\hat{f}^n_{i-\frac{1}{2}})\label{eq:142}.   
  \end{align}

\section{Numerical experiments}

In this section, we provide numerical experiments to demonstrate the effect of our methods. We first test our new WENO approximation on the diffusion term in Section 5.1, and then apply WENO schemes to the complete system of EYM equations in Section 5.2.  Comparison between the first order Lax-Friedrichs scheme and the high-order WENO scheme will be given in Section 5.3.

\subsection{Numerical accuracy test for EYM equations}

In this section, we apply our WENO schemes  (\ref{eq:141})-(\ref{eq:142}) to the complete EYM equations. We consider the computational domain $[-5,5]$ and the following initial boundary conditions: 
\begin{align}\label{eq:bc}
	W(x,0)&=\tanh(10(x-0.1)),\\
	W(-5,t)&=-1,\quad W(5,t)=1,\quad A(-5)=1.
\end{align} 
We march our scheme in time to steady-state at which the maximum absolute point value of $W_t$ reduces to machine zero. The steady-state solutions for $A$ and $W$ using our WENO schemes with 3200 grid points are shown in Fig.~\ref{figWENO}. Here $A$ is Lipschitz continuous. We further test the accuracy of our scheme in the smooth region. Since we do not know the exact solution of EYM equations, we use numerical solutions obtained on a very fine grid with $2^{15}$ points as approximations to the exact solutions. The $L^\infty$ and $L^1$ errors and orders of our scheme are shown in Table \ref{table1}, in which we can clearly observe the expected fourth-order convergence rate. 

 \begin{figure}[!htp]
 \centering
    \subfigure[$A$]{
    \includegraphics[width=12cm]{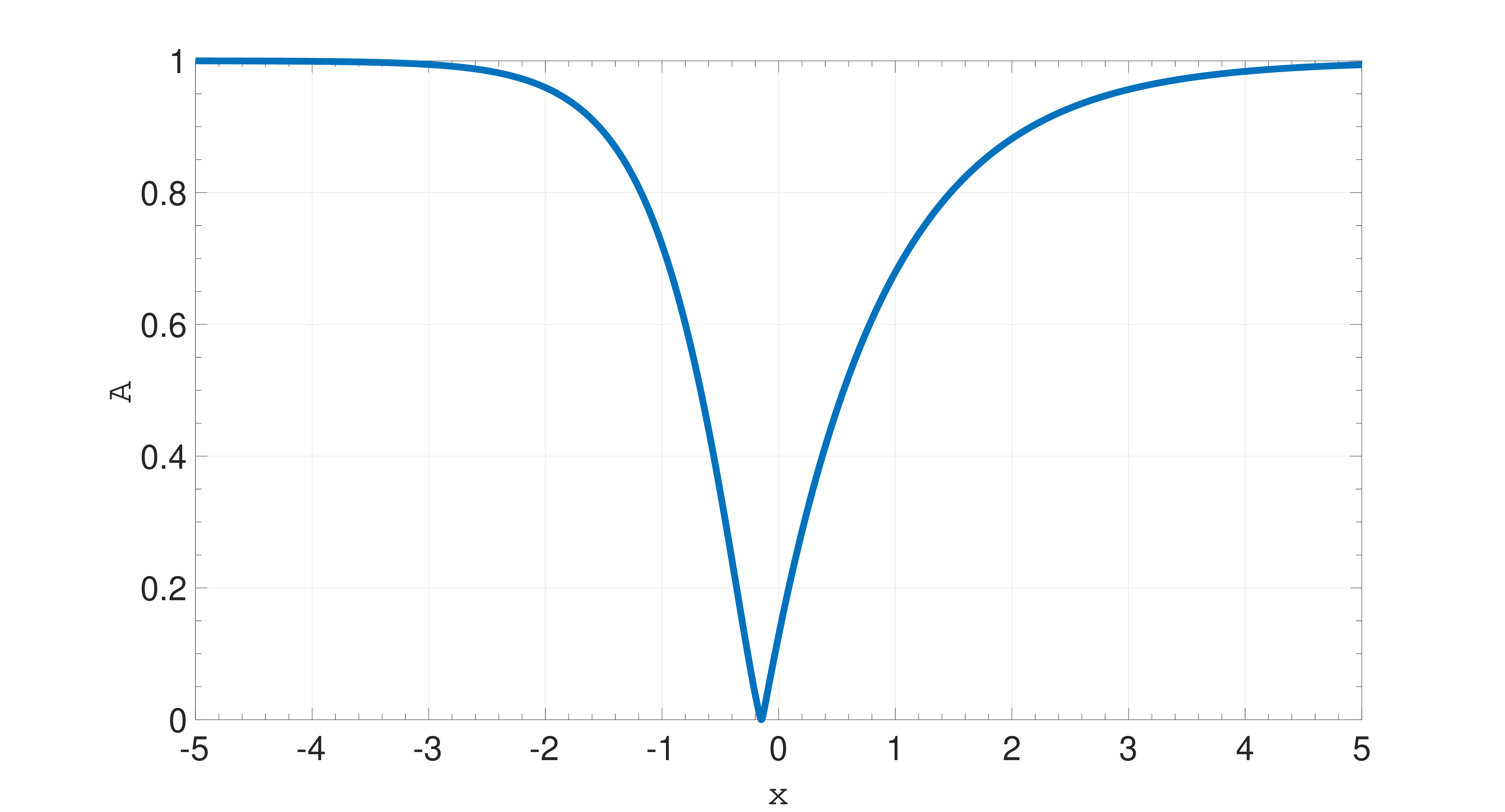} }
    \subfigure[$W$]{
    \includegraphics[width=12cm]{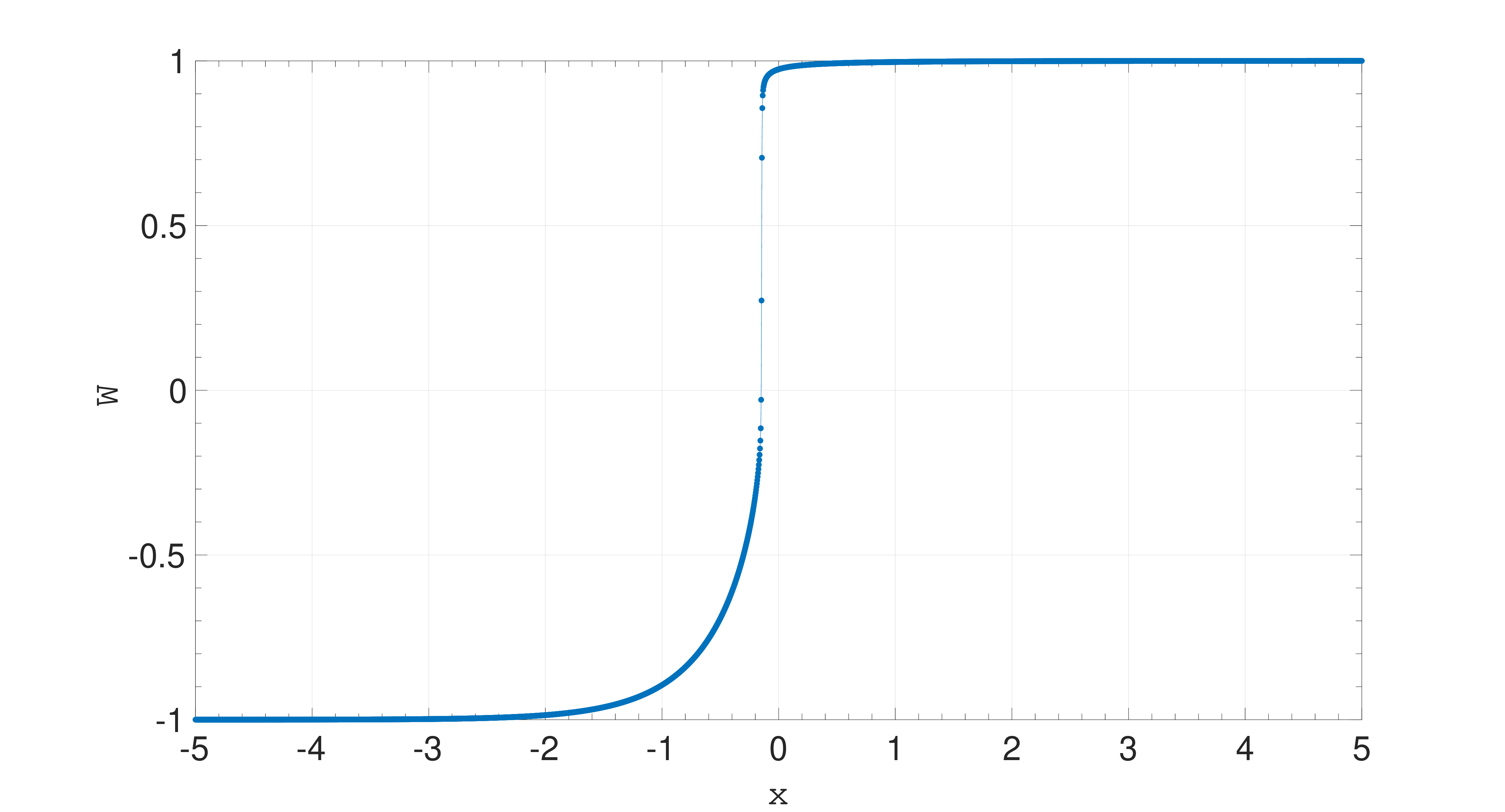} }
    \subfigure[Zoom-in figure of $W$ in the horizon region]{
    \includegraphics[width=12cm]{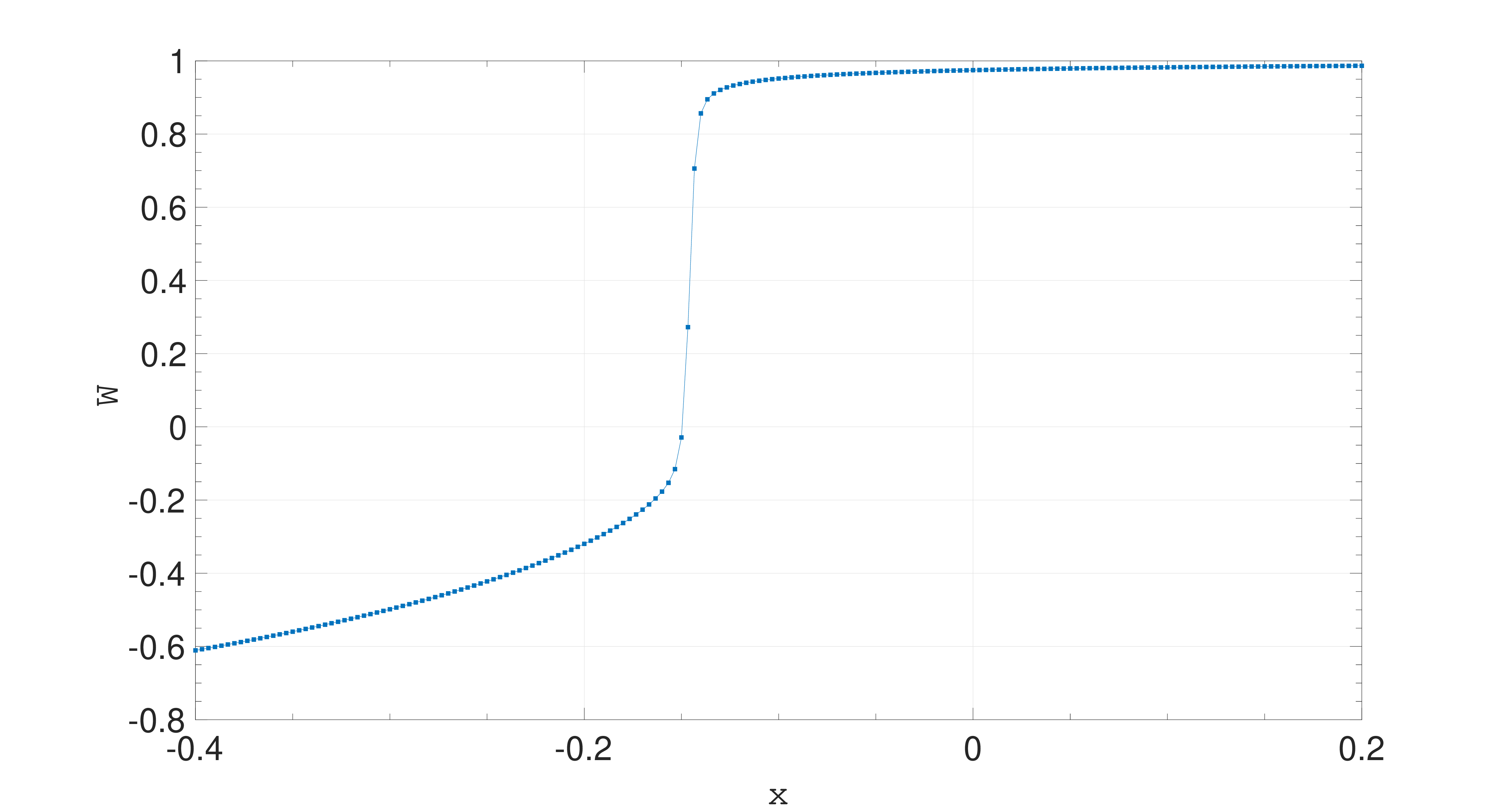} }
 \caption{Solutions to EYM equations solved by WENO schemes with 3200 grid.}
 \label{figWENO}
 \end{figure}

\begin{table*}[!htp]
\centering
\def\arraystretch{2}
\tabcolsep0.3cm
 \begin{tabular}{cccccccccc}
  \hline
  N &&A&&&&&W&&\\
  \cline{2-5} \cline{7-10}
        &$L^1$ error   & Order   &$L^\infty$ error   & Order& &    $L^1$ error  & Order   &  $L^\infty$ error & Order\\
  \hline 
  $100$        &    3.58E-04  &   --    &     1.28E-03   &  --   &&      6.67E-05 & --     &  2.13E-04&  --   \\
  $200$       &    2.61E-05  &    3.78&     9.81E-05  &  3.70&&   5.25E-06   &3.66   &  1.53E-05& 3.79\\
  $400$       &    1.79E-06  &    3.86&     6.63E-06   &  3.88& &  4.07E-07   & 3.68  & 1.16E-06& 3.72\\
  $ 800 $     &    1.19E-07  &    3.91&     4.44E-07   &   3.89&&   2.89E-08  & 3.81 & 7.73E-08&3.90\\
  $ 1600$    &    7.69E-08  &    3.95&    2.89E-08  &   3.94&&   1.93E-09   &  3.91 & 4.96E-09&3.96  \\
  $ 3200$    &   4.59E-09  &     4.06&    1.73E-09  &   4.05&&   1.18E-10   &  4.02 & 2.94E-10&4.07 \\
  \hline 
 \end{tabular}
 %}
 \caption{\label{table1} Error table for the fourth-order WENO schemes in the smooth region of the solution.}
\end{table*}

%\subsection{Numerical results for DG}

\subsection{Lax-Friedrichs schemes}

In this subsection, we compute the EYM equations by using the Lax-Friedrichs schemes (\ref{eq:24}) and (\ref{eq136}) and compare the results of different schemes. We use the same initial boundary conditions as in the last subsection and still take 3200 grid points. The steady-state solution for the Lax-Friedrichs schemes are shown in Fig.~\ref{figLax}. We also compare the Lax-Friedrichs scheme and WENO scheme in Fig.~\ref{figcom}. For the Lax-Friedrichs scheme, we can observe that there are few point values of $A$ being negative, even at the steady state. However, we show in Fig.~\ref{fig9} that as the mesh size $h \to 0$, we have $\min(A) \to 0$. So even there are few negative point values, as the mesh size close to 0, we would have $A\geq 0$, and then $\tilde{A}=A$ globally.

 \begin{figure}[!htp]
 \centering
    \subfigure[$A$]{
    \includegraphics[width=12cm]{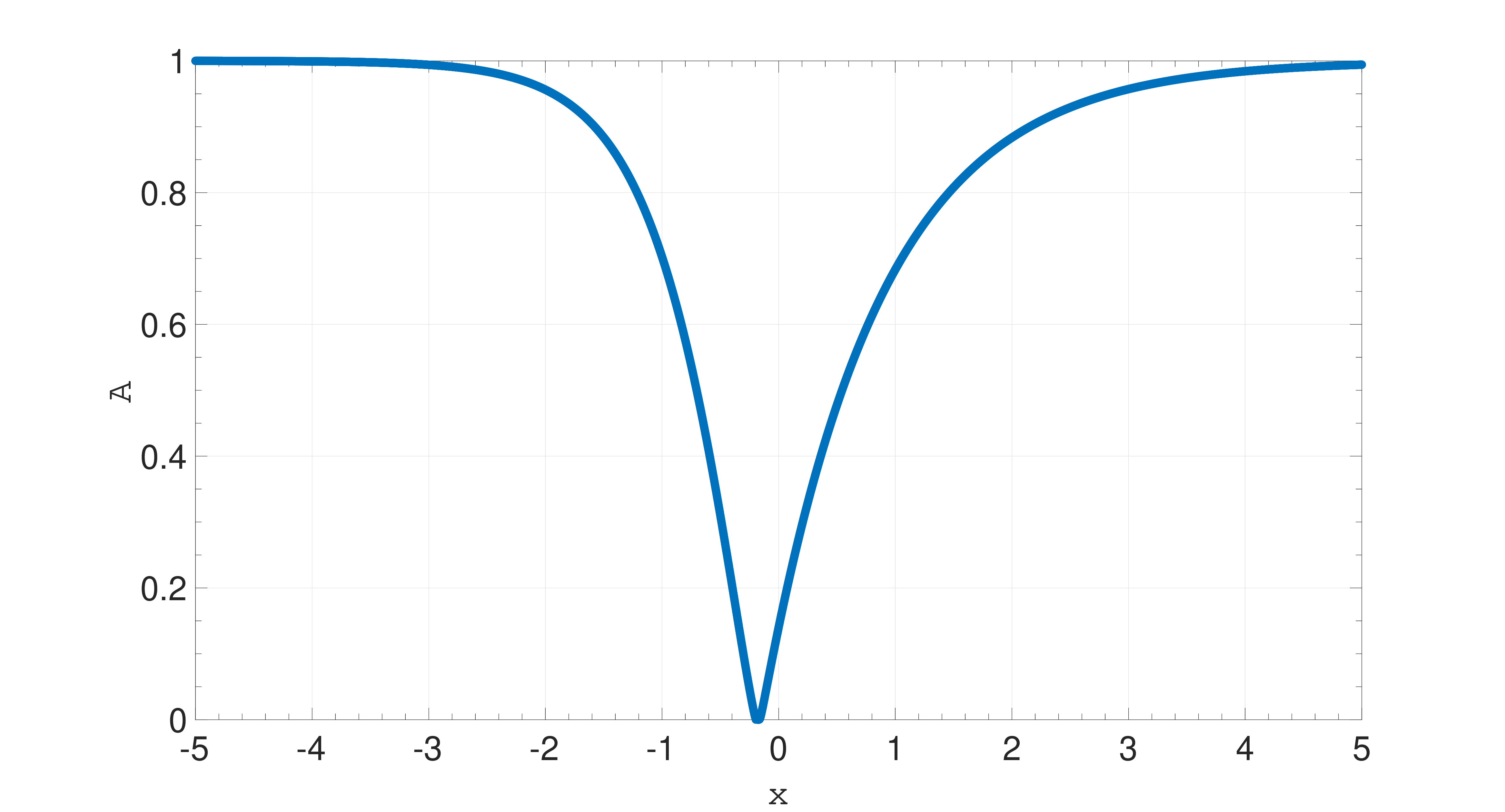} }
    \subfigure[$W$]{
    \includegraphics[width=12cm]{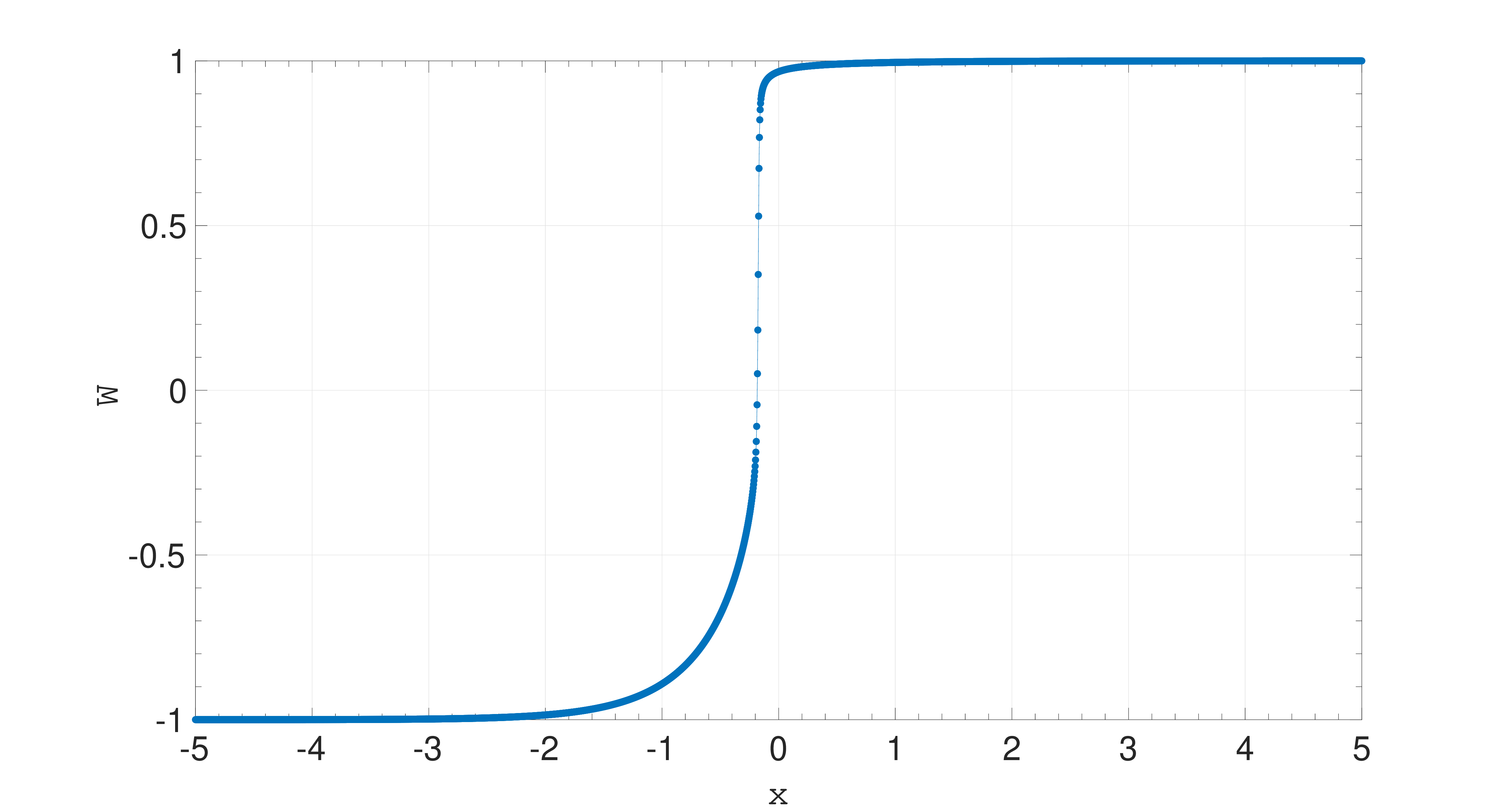} }
 \caption{Solutions to EYM equations solved by Lax Friedrichs schemes with 3200 grid.} \label{figLax}
 \end{figure}

 \begin{figure}[!htp]
 \centering
    \subfigure[$A$]{
    \includegraphics[width=12cm]{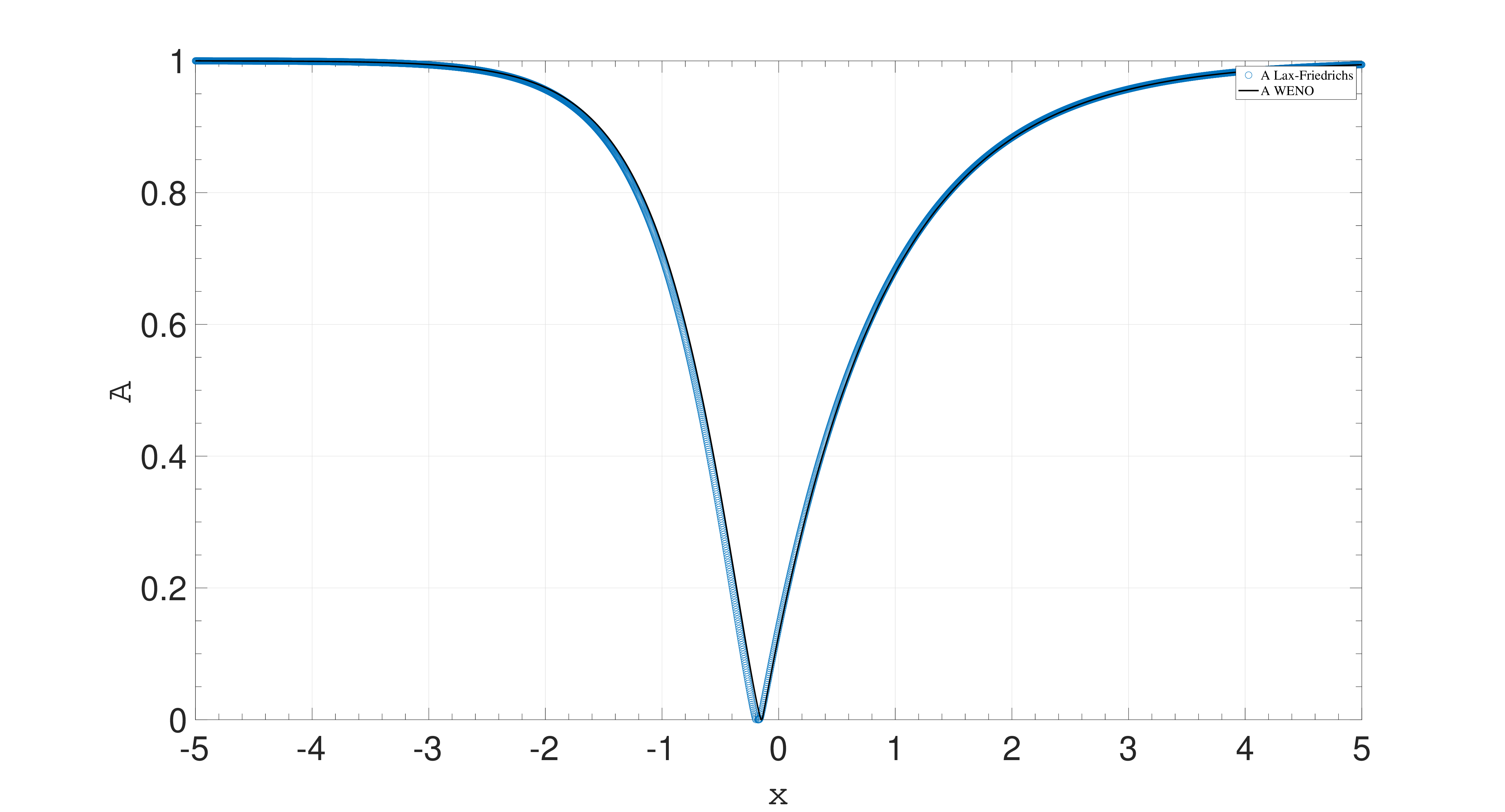} }
    \subfigure[$W$]{
    \includegraphics[width=12cm]{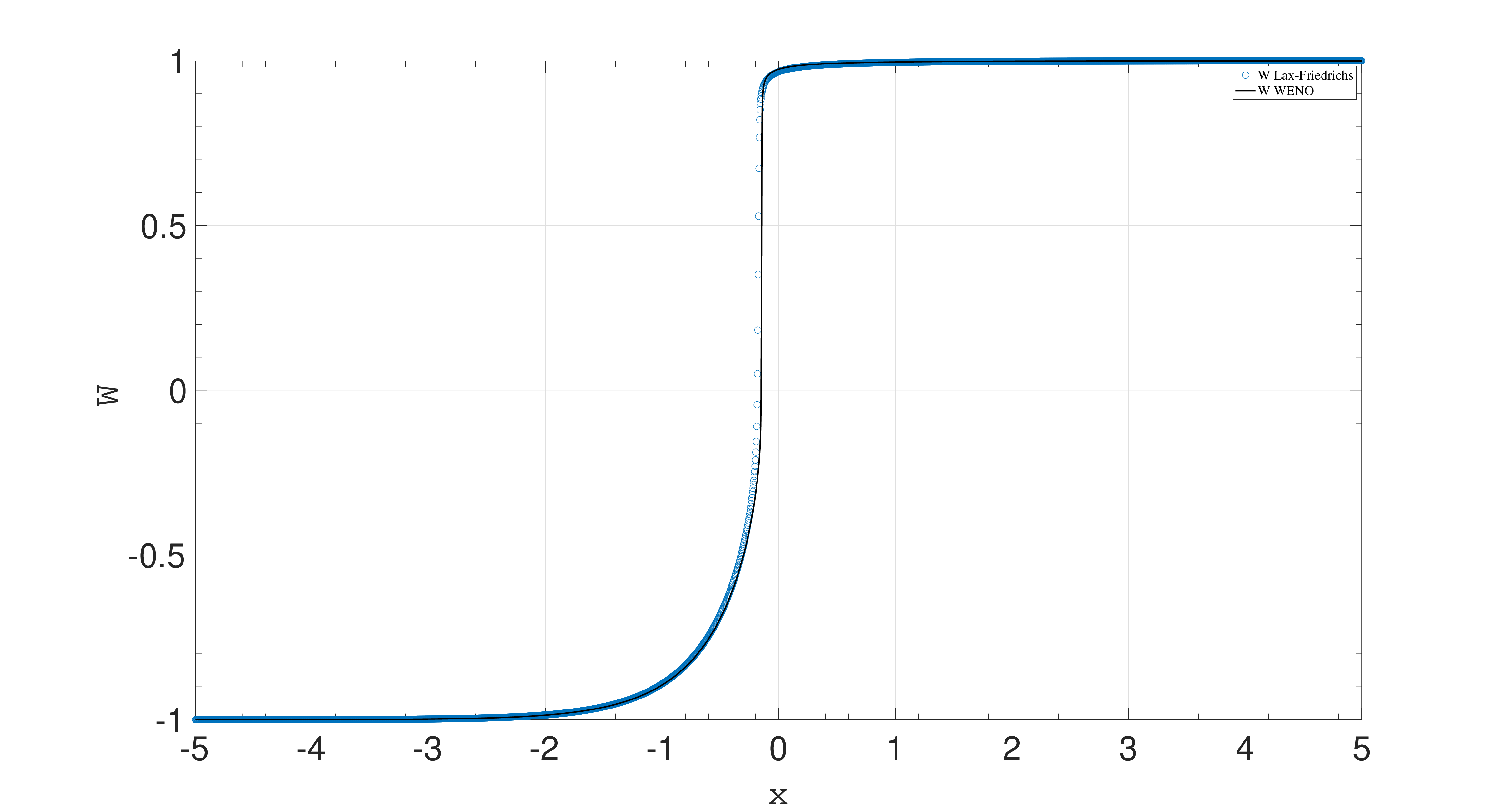} }
 \caption{Comparison of the Lax Friedrichs scheme and the WENO scheme.} \label{figcom}
 \end{figure}

 \begin{figure}[!htp]
	\centering
	  \includegraphics[width=15cm]{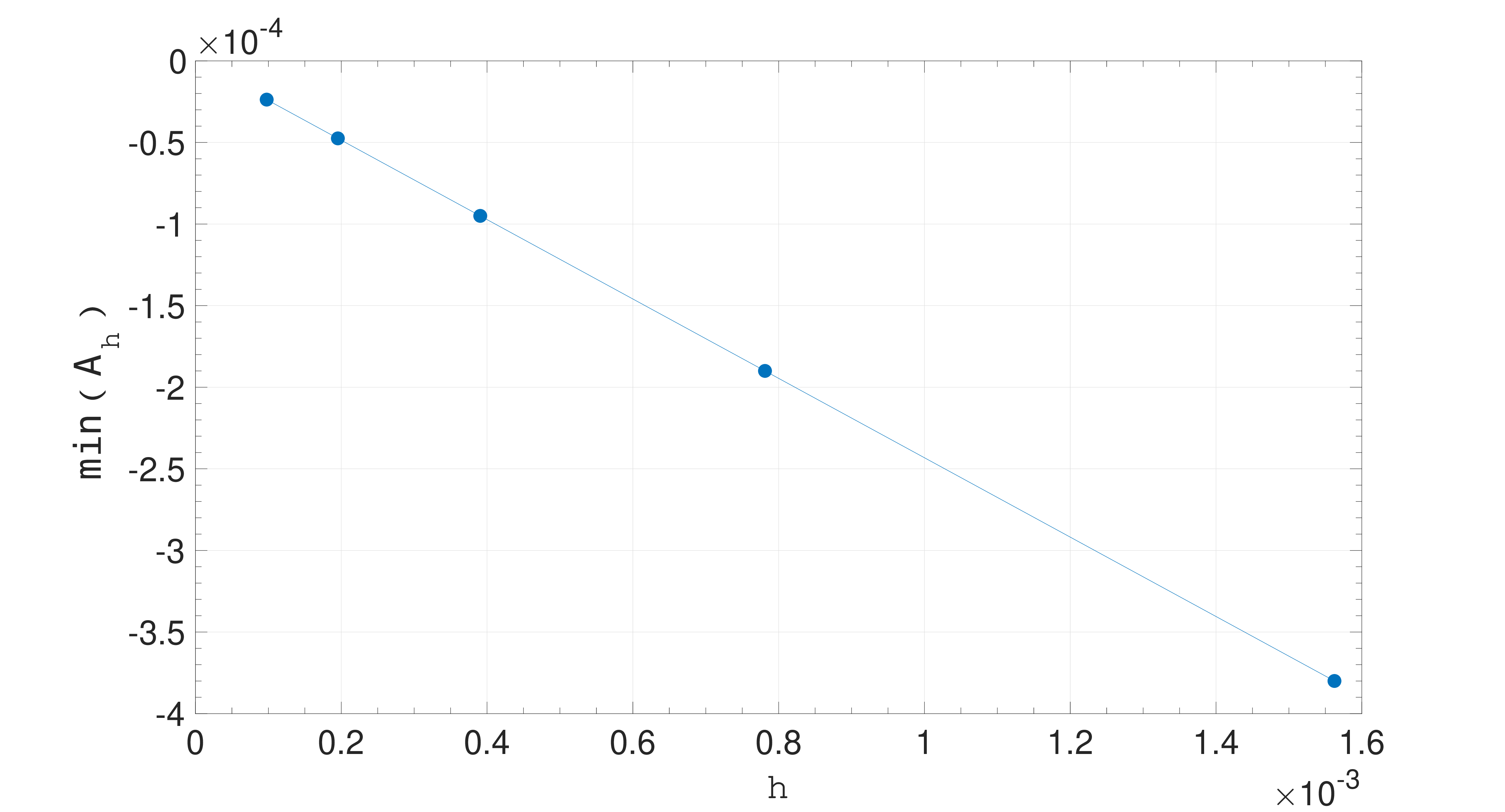}\\
	\caption{The Relationship between $\min(A_h)$ and $h$ in the Lax Friedrichs scheme. As $h \to$ 0, $\min(A_h)=O(h)$. }
	\label{fig9}
\end{figure}
%==================== Jump ==========================

\section{Conclusions}

In this paper, we consider the $SU(2)$ EYM equations and aim to solve for the stable static solutions. We study the first order TVD scheme theoretically and provide new high-order WENO schemes for solving this problem. Numerical experiments are given to show the effect of our schemes.

\end{document}